\documentclass[aps, prx, reprint, amsmath,amssymb,showpacs,floatfix,longbibliography, twocolumn, superscriptaddress,nofootinbib]{revtex4-2}
\usepackage{times} 
\usepackage{graphicx}
\usepackage{CJKutf8}
\usepackage{dcolumn}
\usepackage{bm}
\usepackage{color}
\usepackage{xcolor}
\usepackage{hyperref}
\usepackage{enumitem}
\usepackage{cancel}
\usepackage{amsfonts}
\usepackage{amsthm}
\usepackage{bbm}
\usepackage{xr-hyper}
\usepackage[caption=false]{subfig}
\usepackage{stmaryrd}
\usepackage{yhmath}
\usepackage{makecell}

\usepackage{comment}
\hypersetup{colorlinks=true,citecolor=blue,linkcolor=blue, urlcolor=blue}
\hypersetup{linktocpage}
\newcolumntype{M}[1]{>{\centering\arraybackslash}m{#1}}
\newcolumntype{N}{@{}m{0pt}@{}}
\usepackage{environ}

\bibpunct{[}{]}{,}{n}{}{}

\usepackage{amsfonts,amssymb,amsmath}
\usepackage[T1]{fontenc}
\usepackage{tikz}
\newcommand{\bra}[1]{\langle {#1} |}
\newcommand{\ket}[1]{ | {#1} \rangle}
\newcommand{\braket}[2]{ \langle {#1} | {#2} \rangle}
\let\originalleft\left
\let\originalright\right
\renewcommand{\left}{\mathopen{}\mathclose\bgroup\originalleft}
\renewcommand{\right}{\aftergroup\egroup\originalright}

\newcommand{\T}{\mathrm{T}}

\newcommand{\tr}{\mathrm{tr}}

\newcounter{example}

\newtheorem{theorem}{Theorem}

\NewEnviron{eqs}{%
\begin{equation}\begin{split}
    \BODY
\end{split}\end{equation}
}

\begin{document}
\begin{CJK*}{UTF8}{gbsn}

\title{Approximate quantum error correction theory of non-isometric codes}

\author{Yixu Wang}
\email{wangyixu@simis.cn}
\affiliation{Shanghai Institute for Mathematics and Interdisciplinary Sciences (SIMIS), Shanghai 200433, China}

\author{Yijia Xu (许逸葭)}
\email{yijia@terpmail.umd.edu}
\affiliation{Joint Center for Quantum Information and Computer Science, University of Maryland, College Park,
Maryland 20742, USA}

\author{Zi-Wen Liu}
\email{zwliu0@tsinghua.edu.cn}
\affiliation{Yau Mathematical Sciences Center, Tsinghua University, Beijing, 100084, China }

\date{\today}

\begin{abstract}
Non-isometric encoding arises in various important contexts in quantum error correction, most notably in the finite-energy, non-ideal codewords inevitable in experimental realizations of continuous-variable codes, and holographic quantum gravity. In this work, we present a general and systematic theory of non-isometric quantum error-correcting codes. 
In particular, we employ the approximate quantum error correction framework to quantitatively study the fundamental limitations imposed by non-isometric encodings on the accuracy of quantum error correction and implementation of logical operations.
 We apply our theory to analyze GKP and tiger codes under energy constraints, and discuss the implications to holography.
\end{abstract}

\maketitle
\end{CJK*}

\section{Introduction}
Quantum error correction (QEC) provides the foundational framework for fault-tolerant quantum computation, with QEC codes protecting logical information from physical errors by encoding it into a larger physical Hilbert space. In the standard formulation, the encoding is modeled as an isometric quantum channel, so that the inner product between any two logical states is preserved by the corresponding codewords.

With the rapid development of QEC, non-isometric codes have attracted significant interest due to their natural relevance across practical, physical, and mathematical contexts. 
A ubiquitous source of non-isometric encoding is the non-orthogonality of physical codewords, which can easily arise from practical limitations such as energy constraints, imperfect state preparation, and post-selection. This issue is particularly fundamental to continuous-variable QEC codes.  Taking the well-known cat codes as a simple example~\cite{mirrahimi2014dynamically}, finite-energy constraints lead to coherent state codewords $\ket{\pm\alpha}$ for the orthogonal logical states $\ket{\overline{\pm}}$, which cannot be exactly orthogonal.
Similar non-orthogonality occurs broadly in continuous-variable codes, including Gottesman--Kitaev--Preskill (GKP) codes~\cite{gottesman2000encoding} and tiger codes~\cite{xu2025letting}.  Since the ideal codewords cannot be realized exactly in practice, non-isometric encoding becomes almost inevitable in realistic implementations.
However, the vast majority of existing studies on QEC and fault tolerance fault tolerance assume, by default, perfect isometric encoders and mutually orthogonal codewords.

Beyond practical motivations, a particularly representative fundamental setting where non-isometric encoding plays an important role is quantum gravity, especially in the context of the AdS/CFT correspondence.
The basic reason is that the bulk effective field theory can have more states than can be faithfully represented in the fundamental boundary theory.
Consequently, the bulk-to-boundary map must be non-isometric because it has a nontrivial kernel. Such non-isometric holographic maps have been studied in a variety of settings, including black hole interiors~\cite{kim2020ghost,akers2024black,kar2023non,dewolfe2023non}, de Sitter spacetime~\cite{cao2025overlapping}, cosmology~\cite{antonini2023cosmology}, and closed universes~\cite{akers2025observers}. They become asymptotically isometric in suitable large-$N$ limits~\cite{faulkner2022asymptotically}. The relation between non-isometric encoding and state-dependent reconstruction of bulk operators has also been explored~\cite{papadodimas2014state,akers2024black,antonini2025non}.

These considerations call for a systematic theory for understanding and characterizing non-isometric QEC codes, which remains to be developed. In this work, we address this by establishing an information-theoretic framework for quantitatively analyzing the capabilities and limitations of non-isometric encodings. We cast such encodings as a natural form of approximate quantum error correction, which enables us to characterize how non-isometry limits both central properties of QEC codes---logical information recovery and logical operation implementation. With the general theory in place, we apply it to representative examples from energy-constrained continuous-variable quantum error-correcting codes and holographic duality, illustrating its broad implications across both practical quantum technologies and fundamental physics.

\section{Characterizing non-isometric encodings}
In this work, we describe the encoding by a map 
$V: \mathcal{H}_L \rightarrow \mathcal{H}_P, V=\sum_{i=0}^{d-1} \ket{i_p}\bra{\overline i}$ which maps each logical basis state $\ket{\overline{i}}$ to the corresponding codeword $\ket{i_p}$ in the physical Hilbert space.
Here, $\{\ket{\overline{i}}\}$ form a set of orthonormal basis of the finite $d$-dimensional logical Hilbert space $\mathcal{H}_L=\mathbb{C}^d$, and ${\ket{i_p}}$ denotes the corresponding codewords in the physical Hilbert space $\mathcal{H}_P$. Unlike in standard quantum coding settings, we do not assume that the codewords ${\ket{i_p}}$ are normalized or mutually orthogonal. Therefore, the Gram operator $V^{\dagger}V$ is a positive semi-definite Hermitian operator which may have unequal diagonal entries and nonzero off-diagonal entries. If $\{\ket{i_p}\}$ is an orthonormal set, then $V$ is an isometry and $V^\dagger V=\mathbbm{1}_L$. Throughout this paper, we denote the eigenvalues of $V^\dagger V$ by $\{\lambda_i\}$ and the corresponding normalized eigenvectors by $\{\ket{v_i}\}$. Then
\begin{equation}\label{eq:V}
    V^\dagger V=\sum_{i,j=0}^{d-1} \ket{\overline i}\braket{i_p}{j_p}\bra{\overline j}=\sum_{k=0}^{d-1}\lambda_k\ket{v_k}\bra{v_k}.
\end{equation}

Acting on quantum states, $V(\cdot)V^\dagger$ does not preserve the trace of density matrices in general. To obtain  physical states with proper normalization, for inputs satisfying $\tr(V\rho V^\dagger)>0$ we define the non-isometric encoding as a map 
\begin{equation}\label{eq:nonisoencoding}
    \tilde{\mathcal{E}}(\rho)=\frac{V\rho V^\dagger}{\tr(V \rho V^\dagger)}.
\end{equation}
In general, $\tilde{\mathcal{E}}$ is not a quantum channel because it is not linear. Nevertheless, whenever it is well defined, its output is a valid density matrix. This map arises naturally as a post-selected state of a POVM process. From this definition, we see that rescaling $V$ by a constant does not affect the encoding map. Therefore, to compare and contrast with the isometric case, we can set $\tr(V^\dagger V)=d$ when needed.

To rigorously characterize the quantitative behavior of non-isometric codes, we use the following information-theoretic measures. For two quantum channels $\mathcal{N}$ and $\mathcal{M}$ and a fixed state $\rho$, define
 \begin{eqs}
     F_{\rho}(\mathcal{N},\mathcal{M})\equiv F(\mathcal{N}\otimes \mathrm{id}(\ket{\psi_\rho}\bra{\psi_\rho}),\mathcal{M}\otimes \mathrm{id}(\ket{\psi_\rho}\bra{\psi_\rho})).
 \end{eqs}
where $\ket{\psi_{\rho}}$ is an arbitrary pure dilation of the density matrix $\rho$, and $F$ is the fidelity between two states,
$F(\rho,\sigma)=\tr(\sqrt{\sqrt{\rho}\sigma\sqrt{\rho}})$. 
The channel fidelity is the worst-case fidelity
\begin{equation}
F(\mathcal{N},\mathcal{M})\equiv\min_{\rho} F_{\rho}(\mathcal{N},\mathcal{M}).
\end{equation}
Correspondingly, we define the purified distance
\begin{eqs}
    \epsilon_{\rho}(\mathcal{N},\mathcal{M})\equiv\sqrt{1-F_{\rho}^2(\mathcal{N},\mathcal{M})}.
\end{eqs}
It satisfies the triangle inequality, namely, for any state $\rho$ and quantum channels $\mathcal{N}$, $\mathcal{M}$ and $\mathcal{K}$, 
\begin{equation}\label{eq:epsilontriangleineq}
    \epsilon_{\rho}(\mathcal{N},\mathcal{M})\leq \epsilon_{\rho}(\mathcal{N},\mathcal{K})+\epsilon_{\rho}(\mathcal{K},\mathcal{M}).
\end{equation}
The above definitions and properties extend to completely positive maps as long as for each map $\mathcal{N}$, the corresponding $\mathcal{N}\otimes \mathrm{id}(\ket{\psi_\rho}\bra{\psi_\rho})$ is a valid density matrix. In particular, they apply to $\tilde{\mathcal{E}}$.

We characterize the deviation of $V$ from an isometry through the optimal recovery fidelity of the encoding map $\tilde{\mathcal{E}}$ with no extra noise channel. We define the worst-case and the average-case optimal fidelities as well as the corresponding purified distances as
\begin{eqs}\label{eq:Fmin}
F_{\min}&\equiv\max_{\mathcal{R}}F(\mathcal{R}\circ\tilde{\mathcal{E}},\mathrm{id}),~\epsilon_{\min}\equiv\sqrt{1-F_{\min}^2},
\end{eqs}
\begin{eqs}\label{eq:Favg}
F_{\mathrm{avg}}^2&\equiv\max_{\mathcal{R}}\int_{\mathrm{Haar}} d\psi F_{\ket{\psi}\bra{\psi}}^2(\mathcal{R}\circ\tilde{\mathcal{E}},\mathrm{id}),~\epsilon_{\mathrm{avg}}\equiv\sqrt{1-F^{2}_{\mathrm{avg}}},
\end{eqs}
where $\mathrm{id}$ is the logical identity channel.
In view of the close relation between the Choi state $\ket{\psi_c}\equiv\sum_{i=0}^{d-1}\frac{1}{\sqrt{d}}\ket{\overline{i}}\ket{\overline{i}}$ and the Haar average over pure states, we also define 
\begin{eqs}\label{eq:FChoi}
F_{\mathrm{Choi}}\equiv\max_{\mathcal{R}}F_{\ket{\psi_c}\bra{\psi_c}}(\mathcal{R}\circ\tilde{\mathcal{E}},\mathrm{id}),~
\epsilon_{\mathrm{Choi}}\equiv\sqrt{1-F^{2}_{\mathrm{Choi}}}.
\end{eqs}

A spectral parameter capturing the deviation from isometry is the condition ratio of $V^\dagger V$:
\begin{equation}
r\equiv\frac{\lambda_{\mathrm{min}}}{\lambda_{\mathrm{max}}},~0\leq r\leq 1,
\end{equation}
where $\lambda_{\mathrm{min}}$ and $\lambda_{\mathrm{max}}$ are the minimum and maximum eigenvalues, respectively.
The parameter $r$ takes its maximal value $r=1$ if and only if the encoding is isometric. At the opposite extreme, $r=0$ means $V$ has a nontrivial kernel, so some nonzero logical vector is mapped to zero. This occurs in particular when $\mathrm{dim}(\mathcal{H}_P) < \mathrm{dim}(\mathcal{H}_L)$.

\section{Optimal recovery accuracy}

\begin{figure*}[t]
    \centering
    \includegraphics[width=1.0\textwidth]{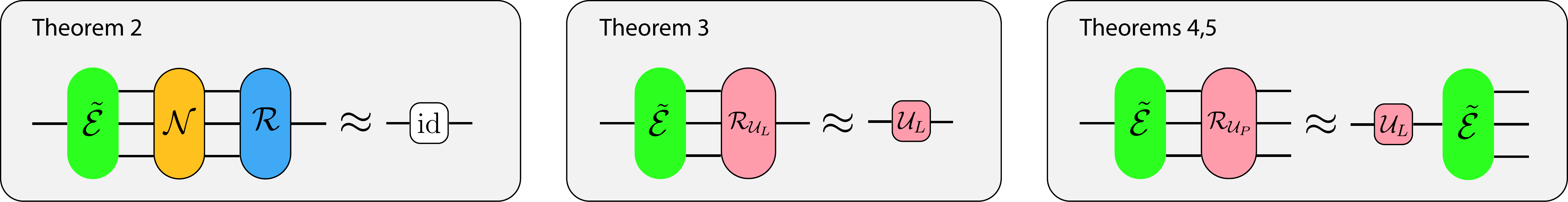}
    \caption{Schematic overview of the approximation relations for non-isometric codes discussed in Theorems~\ref{prop:optimalFnoise}, \ref{prop:logical1}, \ref{prop:logical2}, \ref{prop:bestworstUL}. 
    }
    \label{fig:theoremillustration}
\end{figure*}

We now establish the quantitative consequences of non-isometric quantum encoding. We begin our analysis by deriving the intrinsic optimal recovery fidelity and error quantities of non-isometric encodings without extra noise. We derive explicit formulas for these quantities.
\begin{theorem} \label{prop:optimalF}
    Let $\tilde{\mathcal{E}}$ be the encoding map  $\tilde{\mathcal{E}}(\rho)=\frac{V\rho V^\dagger}{\tr(V \rho V^\dagger)}$ where $V=\sum_{i=0}^{d-1} \ket{i_p}\bra{\overline i}$, and let the eigenvalues and eigenvectors of $V^\dagger V$ be $\{\lambda_i\}$ and $\{\ket{v_i}\}$, respectively. Then
    \begin{align}
&F_{\mathrm{min}}=\frac{2\lambda_{\max}^{\frac 14}\lambda_{\min}^{\frac 14}}{\lambda_{\max}^{\frac 12}+\lambda_{\min}^{\frac 12}}=\frac{2r^{\frac{1}{4}}}{1+r^{\frac{1}{2}}},\\    &F_{\mathrm{avg}}^2=\sum_{i,j=0}^{d-1} \sqrt{\lambda_i\lambda_j} \frac{\partial}{\partial \lambda_i}\frac{\partial}{\partial \lambda_j}\left( \frac{1}{d} \sum_{n=0}^{d-1} \frac{\lambda_n^d \log\lambda_n}{\prod_{m\neq n}(\lambda_n -\lambda_m)} \right),\\
        &F_{\mathrm{Choi}}=\frac{1}{\sqrt{d}}\frac{\sum_{i=0}^{d-1}\lambda^{\frac{1}{2}}_i}{(\sum_{i=0}^{d-1}\lambda_i)^{\frac{1}{2}}}=\frac{1}{\sqrt{d}}\frac{\|V\|_1}{\|V\|_2}.
    \end{align}
\end{theorem}
The detailed derivation is presented in Appendix~\ref{sec:propcalculation}. These fidelities characterize the intrinsic error correction accuracy limits induced by non-isometric encoding.  In deriving these results, we formulate a general theorem extending the information-disturbance theorem in Ref.~\cite{beny2010general} to possibly non-isometric maps. We state and prove this fundamental result in Appendix~\ref{sec:theoremproof}.

Having characterized the intrinsic recovery limits, we now incorporate physical noise that can occur after the encoding. We model the noise by a quantum channel acting on the physical Hilbert space $\mathcal{N}:\mathcal{H}_P\to\mathcal{H}_P$. Analogous fidelity quantities and their corresponding purified distances can be defined: 
\begin{align}
    &F_{\min}^{(\mathcal{N})}\equiv \max_{\mathcal{R}}F(\mathcal{R}\circ\mathcal{N}\circ\tilde{\mathcal{E}},\mathrm{id}),~\epsilon_{\min}^{(\mathcal{N})}\equiv\sqrt{1-F_{\min}^{(\mathcal{N})2}},\\
    &F_{\mathrm{Choi}}^{(\mathcal{N})}\equiv\max_{\mathcal{R}}F_{\ket{\psi_c}\bra{\psi_c}}(\mathcal{R}\circ\mathcal{N}\circ\tilde{\mathcal{E}},\mathrm{id}),~\epsilon_{\mathrm{Choi}}^{(\mathcal{N})}\equiv\sqrt{1-F_{\mathrm{Choi}}^{(\mathcal{N})2}},\\
    &F_{\mathrm{avg}}^{(\mathcal{N})2}\equiv\max_{\mathcal{R}}\int_{\mathrm{Haar}} d\psi F_{\ket{\psi}\bra{\psi}}^2(\mathcal{R}\circ\mathcal{N}\circ\tilde{\mathcal{E}},\mathrm{id}),\\
    &\epsilon_{\mathrm{avg}}^{(\mathcal{N})}\equiv\sqrt{1-F_{\mathrm{avg}}^{(\mathcal{N})2}}.
\end{align}
 
When $\tilde{\mathcal{E}}$ is a quantum channel, it is known that $\epsilon_{\mathrm{Choi}}^{(\mathcal{N})}=\sqrt{\frac{d+1}{d}}\epsilon_{\mathrm{avg}}^{(\mathcal{N})}$~\cite{gilchrist2005distance, horodecki1999general,Kong2022NearOptimal}.  In the non-isometric setting of interest here, this equivalence relation no longer holds. Nevertheless, the Choi error remains an upper bound on the average error, $\epsilon_{\mathrm{Choi}}^{(\mathcal{N})}\geq\epsilon_{\mathrm{avg}}^{(\mathcal{N})}$, as proved in Appendix~\ref{sec:avgvschoi}.

The following theorem delineates the fundamental approximate error-correction properties of non-isometric encodings in the presence of noise, organized into three layers: intrinsic lower bounds, noise-dependent upper bounds, and optimality conditions. First, the noisy recovery error is always lower-bounded by the corresponding intrinsic error derived in Theorem~\ref{prop:optimalF}. Second, it sets an upper bound consisting of two contributions: the intrinsic non-isometric error and the additional error of correcting the physical noise relative to the encoded family. Finally, it identifies a sufficient condition under which the lower bound is saturated and 
 give its equivalent formulation in the form of Knill--Laflamme (KL) error correction condition \cite{kl1997}. Formally:
\begin{theorem}\label{prop:optimalFnoise}
Let $\tilde{\mathcal{E}}$ be the non-isometric encoding map
$\tilde{\mathcal{E}}(\rho)=\frac{V\rho V^\dagger}{\tr(V \rho V^\dagger)}$ where $V=\sum_{i=0}^{d-1} \ket{i_p}\bra{\overline i}$. Let $\mathcal{N}$ be a noise channel on the physical Hilbert space, described by $\mathcal{N}(\cdot)=\sum_j N_j(\cdot)N_j^\dagger$. Let $\mathcal R_{\mathcal N}$ be a quantum channel on $\mathcal{H}_P$. Then
\begin{enumerate}
\item[(a)] $\epsilon_{\mathrm{min}}\leq\epsilon_{\mathrm{min}}^{(\mathcal{N})} \leq \epsilon_{\mathrm{min}} + \min_{\mathcal{R}_{\mathcal{N}}}\epsilon_{\mathrm{min}}(\mathcal R_{\mathcal N}\circ\mathcal{N}\circ \tilde{\mathcal{E}},\tilde{\mathcal{E}});$

\item[(b)] $\epsilon_{\mathrm{min}}=\epsilon_{\mathrm{min}}^{(\mathcal{N})}$ if $\min_{\mathcal{R}_{\mathcal{N}}}\epsilon_{\mathrm{min}}(\mathcal R_{\mathcal N}\circ\mathcal{N}\circ \tilde{\mathcal{E}},\tilde{\mathcal{E}})=0;$

\item[(c)] $\min_{\mathcal{R}_{\mathcal{N}}}\epsilon_{\mathrm{min}}(\mathcal R_{\mathcal N}\circ\mathcal{N}\circ \tilde{\mathcal{E}},\tilde{\mathcal{E}})=0$ if and only if there exists a matrix $\Lambda=(\Lambda_{ij})$ such that $V^\dagger N^\dagger_i N_j V=\Lambda_{ij} V^\dagger V,~ \forall~i,j.$
\end{enumerate}
\noindent The above properties apply directly to the corresponding Choi and average quantities.
\end{theorem}

The setup of this theorem is illustrated in the left panel of Figure \ref{fig:theoremillustration}. The proof is given in Appendix~\ref{sec:optimalFineqproof}. 
Note that part (b) provides a sufficient condition for saturating the lower bound, but it is not necessary in general unless $\epsilon_{\mathrm{min}}=0$. Result (c) gives the non-isometric analogue of the KL condition. Indeed, when $V$ is an isometry, this condition reduces to the conventional KL condition~\cite{beny2010general,kl1997}.

\section{No exact logical unitary operator}
Another central consequence of non-isometric encoding concerns logical operations, a fundamental problem in quantum coding.
Beyond optimal recovery, we now demonstrate that the implementation of logical unitary operations on the physical system is also subject to fundamental accuracy limits, revealing another manifestation of the obstruction induced by non-isometry.  Note that the setup and notations of all theorems are illustrated in Figure \ref{fig:theoremillustration}.

In the ordinary isometric setting, let the encoding isometry be $V_0$, satisfying $V^\dagger_0V_0=\mathbbm{1}_L$ and $V_0V_0^\dagger=P_V$. A physical unitary $U_P$ represents a logical unitary $U_L$ if 
    $U_P V_0\ket{\psi}= V_0 U_L \ket{\psi},~\forall \ket{\psi}\in \mathcal{H}_L$.
For finite-dimensional logical Hilbert spaces, the following two formulations in terms of operator identities are equivalent, as shown in Appendix~\ref{sec:callogical}:
\begin{eqs}\label{eq:logicalequivalence}
   V_0^\dagger U_P V_0= U_L \quad \Leftrightarrow \quad U_P V_0= V_0 U_L.
\end{eqs} 
Crucially, although these two perspectives are equivalent for isometric encodings, they yield different characterizations of logical operators in the non-isometric case. The left-hand identity describes a recovery criterion: one asks whether a logical unitary channel can be realized by first applying the non-isometric encoding and then applying an appropriate recovery channel. Under this perspective, we show that logical unitaries generally cannot be implemented exactly, and further, the optimal implementation fidelity is exactly characterized by the code fidelity under different measures.
\begin{theorem}\label{prop:logical1}
    Let $\tilde{\mathcal{E}}$ be the non-isometric encoding map $\tilde{\mathcal{E}}(\rho)=\frac{V\rho V^\dagger}{\tr(V \rho V^\dagger)}$ where $V=\sum_{i=0}^{d-1} \ket{i_p}\bra{\overline i}$. Let $\mathcal{U}_L(\cdot)=U_L (\cdot)U_L^\dagger$ be a logical unitary channel.  Then there does not exist any quantum channel $\mathcal{R}_{\mathcal{U}_L}: \mathcal{H}_P \rightarrow \mathcal{H}_L$ 
    such that  $\mathcal{R}_{\mathcal{U}_L}\circ\tilde{\mathcal{E}}=\mathcal{U}_L$. Moreover, the optimal worst-case and average channel fidelities are given by $F_{\mathrm{min}}$, $F_{\mathrm{avg}}^2$ and  $F_{\mathrm{Choi}}$, respectively:
    \begin{align}
    &\max_{\mathcal{R}_{\mathcal{U}_L}}F(\mathcal{R}_{\mathcal{U}_L}\circ\tilde{\mathcal{E}},\mathcal{U}_L)=F_{\mathrm{min}},\\
    &\max_{\mathcal{R}_{\mathcal{U}_L}}F_{\ket{\psi_c}\bra{\psi_c}}(\mathcal{R}_{\mathcal{U}_L}\circ\tilde{\mathcal{E}},\mathcal{U}_L)=F_{\mathrm{Choi}},\\
    &\max_{\mathcal{R}_{\mathcal{U}_L}}\int_{\mathrm{Haar}} d\psi F_{\ket{\psi}\bra{\psi}}^2(\mathcal{R}_{\mathcal{U}_L}\circ\tilde{\mathcal{E}},\mathcal{U}_L)=F_{\mathrm{avg}}^2.    
    \end{align}
\end{theorem}
\begin{proof}
We use Theorem~\ref{theorem:fidelitycomplementary} for the case $(n_1,n_2)=(1,0)$, $\mathcal{N}_1=\mathrm{id}_L$, $\mathcal{N}_2=\mathrm{id}_P$, $\mathcal{M}_1=\mathcal{U}_L$, $\mathcal{M}_2=\mathrm{id}_L$. Because the complementary channel for both the unitary channel and the identity channel is the replacement channel $(\mathcal{U}_L)_c(\rho)=(\mathrm{id}_L)_c(\rho)=\tr(\rho)\ket{0}\bra{0}$, the result is immediate for the worst-case and Choi fidelities. The average fidelity result follows similarly from the invariance of fidelity under unitary conjugation.
\end{proof}

The right-hand identity in \eqref{eq:logicalequivalence} expresses the usual physical implementation or covariance perspective: applying a physical unitary after encoding is equivalent to applying the logical unitary before encoding.
This conventional interpretation is adopted in previous work analyzing logical unitary operators in the context of non-isometric encoding \cite{antonini2025non}. 
For non-isometric encodings, the achievable fidelity under this criterion depends nontrivially on the chosen logical unitary. The following theorem gives an explicit characterization.
\begin{theorem}\label{prop:logical2}
Let $\tilde{\mathcal{E}}:\mathcal{H}_L\to\mathcal{H}_P$ be the non-isometric encoding map $\tilde{\mathcal{E}}(\rho)=\frac{V\rho V^\dagger}{\tr(V \rho V^\dagger)}$ where $V=\sum_{i=0}^{d-1} \ket{i_p}\bra{\overline i}$. 
Let $\mathcal{U}_L(\cdot)=U_L (\cdot)U_L^\dagger$ be a logical unitary channel.
Then for any state $\rho\in\mathcal{H}_L$, we have
 \begin{align}\label{eq:fidelityUL}
&\max_{\mathcal{R}_{\mathcal{U}_P}}F_{\rho}(\mathcal{R}_{\mathcal{U}_P}\circ\tilde{\mathcal{E}},\tilde{\mathcal{E}}\circ\mathcal{U}_L)\notag\\=&\frac{\tr\left[\sqrt{\sqrt{\rho^{1/2}V^\dagger V \rho^{1/2}} \rho^{1/2}U_L^\dagger V^\dagger V U_L\rho^{1/2}\sqrt{\rho^{1/2}V^\dagger V \rho^{1/2}}}\right]}{\sqrt{\tr(\rho V^\dagger V)\tr(\rho U_L^\dagger V^\dagger V U_L)}}. 
 \end{align}
Here, the maximization is taken over all physical channels  $\mathcal{R}_{\mathcal{U}_P}$. The maximization can be achieved by a unitary channel $\mathcal{R}_{\mathcal{U}_P}^*(\cdot)=U_P(\cdot)U_P^\dagger$. Here $U_P$ is block-diagonal. On the subspace given by the support of $VV^\dagger$, it is equivalent to the partial isometry in the polar decomposition of $V U_L \rho V^\dagger$, while on the orthogonal complement, it can be chosen arbitrarily.
\end{theorem} 

 The proof is given in Appendix~\ref{sec:callogical2}. 
Evidently, for any pure state $\ket{\psi}\in\mathcal{H}_L$, we have $\max_{\mathcal{R}_{\mathcal{U}_P}}F_{\ket{\psi}\bra{\psi}}(\mathcal{R}_{\mathcal{U}_P}\circ\tilde{\mathcal{E}},\tilde{\mathcal{E}}\circ\mathcal{U}_L)=1$. Indeed, if we only require exact implementation $U_P V\ket{\psi}=V U_L\ket{\psi}$ for a single state $\ket{\psi}$, one can always construct such a $U_P$. This is related to the state-dependent reconstruction in the context of holographic duality~\cite{papadodimas2014state,akers2024black,kar2023non,antonini2025non}. 

If we take $\rho=\rho_{\mathrm{Choi}}=\frac{\mathbbm{1}_L}{d}$ in Eq.~\eqref{eq:fidelityUL}, we obtain the quantity $\max_{\mathcal{R}_{\mathcal{U}_P}}F_{\mathrm{Choi}}(\mathcal{R}_{\mathcal{U}_P}\circ\tilde{\mathcal{E}},\tilde{\mathcal{E}}\circ\mathcal{U}_L)$, 
which characterizes, for a certain logical unitary $U_L$, the extent to which it can be represented by a single physical channel on $\mathcal {H}_P$ in a state-independent manner:
\begin{align}\label{eq:FChoilogical}
&\max_{\mathcal{R}_{\mathcal{U}_P}}F_{\mathrm{Choi}}(\mathcal{R}_{\mathcal{U}_P}\circ\tilde{\mathcal{E}},\tilde{\mathcal{E}}\circ\mathcal{U}_L)\nonumber\\&\qquad=\frac{\tr\left[\sqrt{\sqrt{V^\dagger V } U_L^\dagger V^\dagger V U_L\sqrt{V^\dagger V }}\right]}{\tr(V^\dagger V)}.
\end{align} 

With the above understanding that non-isometry affects the implementation accuracy of different logical unitaries non-uniformly, it is natural to ask which logical unitaries are more strongly or more weakly obstructed. The following theorem gives a complete explicit characterization of the best (exact implementation) and worst cases in terms of the spectral properties of $V^\dagger V$.

\begin{theorem}\label{prop:bestworstUL}
Let $\tilde{\mathcal{E}}:\mathcal{H}_L\to\mathcal{H}_P$ be the non-isometric encoding $\tilde{\mathcal{E}}(\rho)=\frac{V\rho V^\dagger}{\tr(V \rho V^\dagger)}$, where $V=\sum_{i=0}^{d-1} \ket{i_p}\bra{\overline i}$. For a logical unitary $\mathcal{U}_L(\cdot)=U_L (\cdot)U_L^\dagger$, consider the physical implementation Choi fidelity $\max_{\mathcal{R}_{\mathcal{U}_P}}F_{\mathrm{Choi}}(\mathcal{R}_{\mathcal{U}_P}\circ\tilde{\mathcal{E}},\tilde{\mathcal{E}}\circ\mathcal{U}_L)$, where the maximization is taken over arbitrary physical quantum channels.

On the one hand,
\begin{align}
\max_{U_L}\max_{\mathcal{R}_{\mathcal{U}_P}}F_{\mathrm{Choi}}(\mathcal{R}_{\mathcal{U}_P}\circ\tilde{\mathcal{E}},\tilde{\mathcal{E}}\circ\mathcal{U}_L)=1.
\end{align}
The unitaries $U_{L,\mathrm{max}}$ that achieve the maximum are precisely those that commute with $V^\dagger V$. Equivalently, if $\ket{v_j}$ is a nondegenerate eigenvector of $V^\dagger V$, then
\begin{align}
    U_{L,\mathrm{max}}\ket{v_j}=e^{i\theta_j}\ket{v_j},
\end{align}
where $e^{i\theta_j}$ is an arbitrary phase. Within each degenerate eigenspace of $V^\dagger V$, $U_{L,\mathrm{max}}$ can be an arbitrary unitary.

On the other hand,
\begin{align}
&\min_{U_L}\max_{\mathcal{R}_{\mathcal{U}_P}}F_{\mathrm{Choi}}(\mathcal{R}_{\mathcal{U}_P}\circ\tilde{\mathcal{E}},\tilde{\mathcal{E}}\circ\mathcal{U}_L)=\frac{\sum_{i=0}^{d-1} \sqrt{\lambda_{(i)} \lambda_{(d-1-i)}}}{\sum_{i=0}^{d-1}\lambda_{i}},
\end{align}
where $\lambda_{(0)}\geq\lambda_{(1)}\geq\cdots\geq\lambda_{(d-1)}$ is an ordered sequence of the eigenvalues of $V^\dagger V$. A unitary $U_{L,\mathrm{min}}$ achieving the minimal fidelity can be be written as
\begin{align}
U_{L,\mathrm{min}}\ket{v_{(j)}}=e^{i\theta_j}\ket{v_{(d-1-j)}},
\end{align}
where $\ket{v_{(j)}}$ is an eigenvector associated with $\lambda_{(j)}$ and $e^{i\theta_j}$ is an arbitrary phase.
\end{theorem}
The proof is given in Appendix~\ref{sec:calbestworstUL}. In Appendix~\ref{sec:comparison}, we also compare the quantities characterizing non-isometric properties in this work to the related quantities proposed in Ref.~\cite{antonini2025non}.

\section{Key Applications}
After developing the general framework, we now examine various important scenarios in which non-isometric quantum codes arise, illustrating the utility and versatility of our theory.

\subsection{Approximate GKP codes}
The well-known GKP code \cite{gottesman2000encoding} uses bosonic modes to encode logical qubits. Here we consider a single-mode GKP code encoding a single qubit. The ideal codewords are superpositions of infinitely many Dirac delta functions, which are non-normalizable. Explicitly, the ideal GKP codewords are given by
\begin{eqs}
    \ket{i}_{\textrm{ideal}}=\sum_{k\in\mathbb{Z}} \ket{(2k+i)\sqrt{\pi}}, ~i=0,1.
\end{eqs}
To obtain physical codewords, a certain form of regularization should be applied \cite{tzitrin2020progress,eaton2019non,brenner2025composable}. Here we regularize the ideal codeword by applying the $\beta$-damping operator $\hat{R}=e^{-\beta \hat{n}}$ with $\beta >0$, as proposed in Refs.~\cite{menicucci2014fault,jafarzadeh2025logical}, giving the normalized approximate codewords
\begin{align}
&\ket{i_p}=\sum_{k\in\mathbb{Z}} \hat{R}\ket{(2k+i)\sqrt{\pi}}, ~i=0,1;\\
& V^\dagger V=\sum_{i,j=0}^1 \langle i_p\ket{j_p} \ket{\overline{i}}\bra{\overline{j}},\\
& \langle i_p\ket{j_p} =\sum_{k,l \in \mathbb{Z}} \bra{(2k+i)\sqrt{\pi}} \hat{R}^2 \ket{(2l+j)\sqrt{\pi}}.
\end{align}
For two positional eigenstates $\ket{x}$ and $\ket{y}$, it can be explicitly calculated that
\begin{eqs}
    \bra{y}e^{-\beta \hat{n}}\ket{x}=\frac{1}{\sqrt{\pi (1-e^{-2\beta})}} \exp{\left(-\frac{x^2+y^2}{2\tanh{\beta}}+\frac{xy}{\sinh{\beta}}\right)}.
\end{eqs}
Although the codewords are approximate, the encoding map still satisfies the following exact identities: 
\begin{align}\label{eq:aGKPexact}
   &\langle 0_p\ket{0_p}-\langle 1_p\ket{1_p}=\langle 0_p\ket{1_p}+\langle 1_p\ket{0_p},~\forall \beta;\\
   & \lim_{\beta\to\infty} V^\dagger V =\begin{pmatrix}
       1+\frac{1}{\sqrt{2}}& \frac{1}{\sqrt{2}}\\
       \frac{1}{\sqrt{2}}& 1-\frac{1}{\sqrt{2}}
   \end{pmatrix}. \label{eq:aGKPexact2}
\end{align}
The eigenvectors of $V^\dagger V$ are 
\begin{eqs}
    \ket{v_0}= \frac{\ket{\overline{0}}+(\sqrt{2}-1)\ket{\overline{1}}}{\sqrt{4-2\sqrt{2}}}, ~ \ket{v_1}=\frac{(\sqrt{2}-1)\ket{\overline{0}}-\ket{\overline{1}}}{\sqrt{4-2\sqrt{2}}}.
\end{eqs}
As $\beta$ increases, this non-isometric encoder suppresses the relative amplitude of $\ket{v_1}$ with respect to $\ket{v_0}$ when encoding a generic logical state $\ket{\overline{\psi}}=c_0 \ket{v_0}+c_1\ket{v_1}$.  From \eqref{eq:aGKPexact2}, we see that as $\beta\to\infty$, $V^\dagger V$ becomes proportional to projection $\ket{v_0}\bra{v_0}$, eliminating the $\ket{v_1}$ component. 

With the explicit form of the codewords and $V^\dagger V$, we can directly calculate the optimal noiseless recovery fidelities. In the upper panel of Figure \ref{fig:Fmac}, we plot $F_{\mathrm{min}}$, $F_{\mathrm{avg}}$ and $F_{\mathrm{Choi}}$ with respect to the regularization parameter $\beta$. When $\beta\to0$, all optimal fidelities approach $1$. When $\beta\to\infty$, one of the eigenvalues of $V^\dagger V$ becomes $0$, and consequently  $F_{\mathrm{min}}\to 0$ while $F_{\mathrm{Choi}}=F_{\mathrm{avg}}=1/\sqrt{2}$ in this limit.

Using Eq.~\eqref{eq:FChoilogical}, we can further characterize how well each logical unitary operation can be represented on the physical system. In the lower panel of Figure \ref{fig:Fmac}, we plot the quantity $\max_{\mathcal{R}_{\mathcal{U}_P}}F_{\mathrm{Choi}}(\mathcal{R}_{\mathcal{U}_P}\circ\tilde{\mathcal{E}},\tilde{\mathcal{E}}\circ\mathcal{U}_L)$ for different gates. The first identity in \eqref{eq:aGKPexact} ensures that the Hadamard gate $H$ commutes with $V^\dagger V$ for any value of $\beta$. Hence, $H$ can be represented exactly on the physical Hilbert space.

\begin{figure}
    \centering
    \includegraphics[width=0.7\linewidth]
    {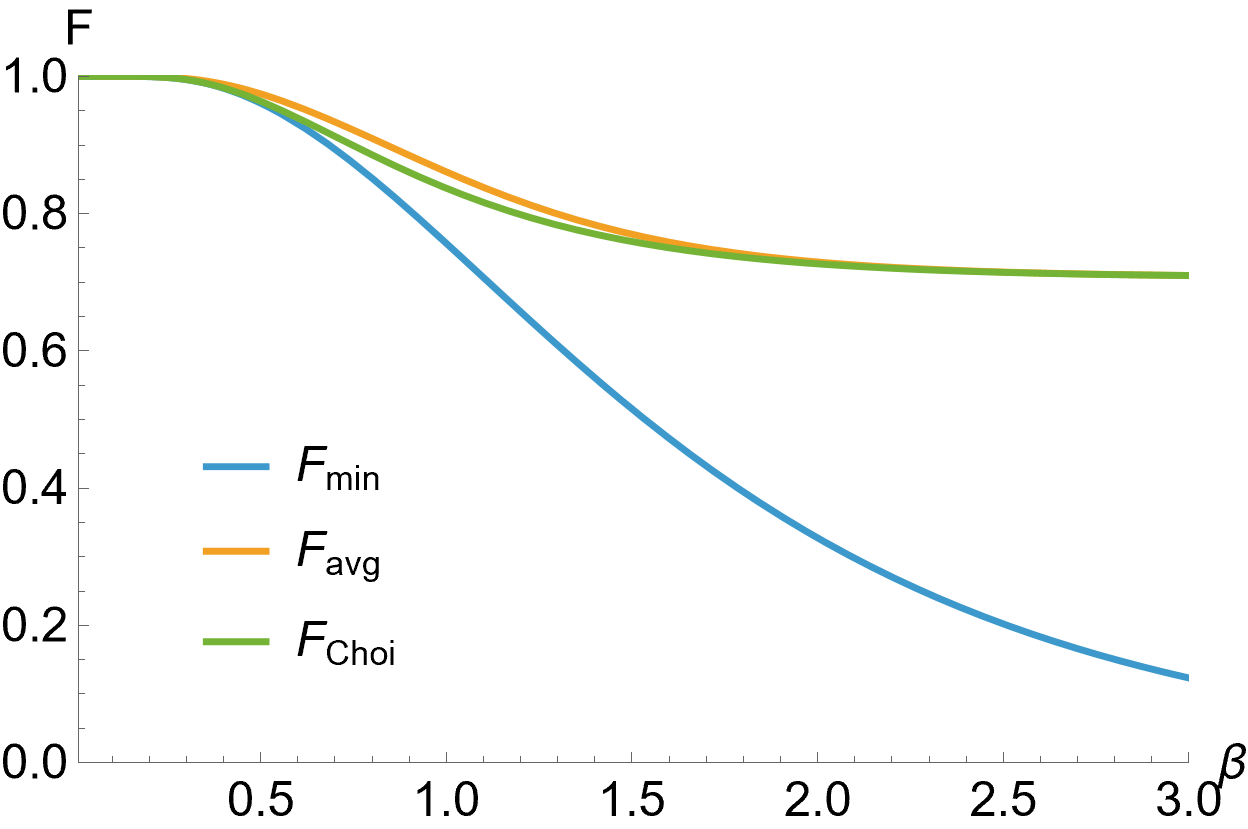}\\
    ~\\
    \includegraphics[width=0.7\linewidth]
    {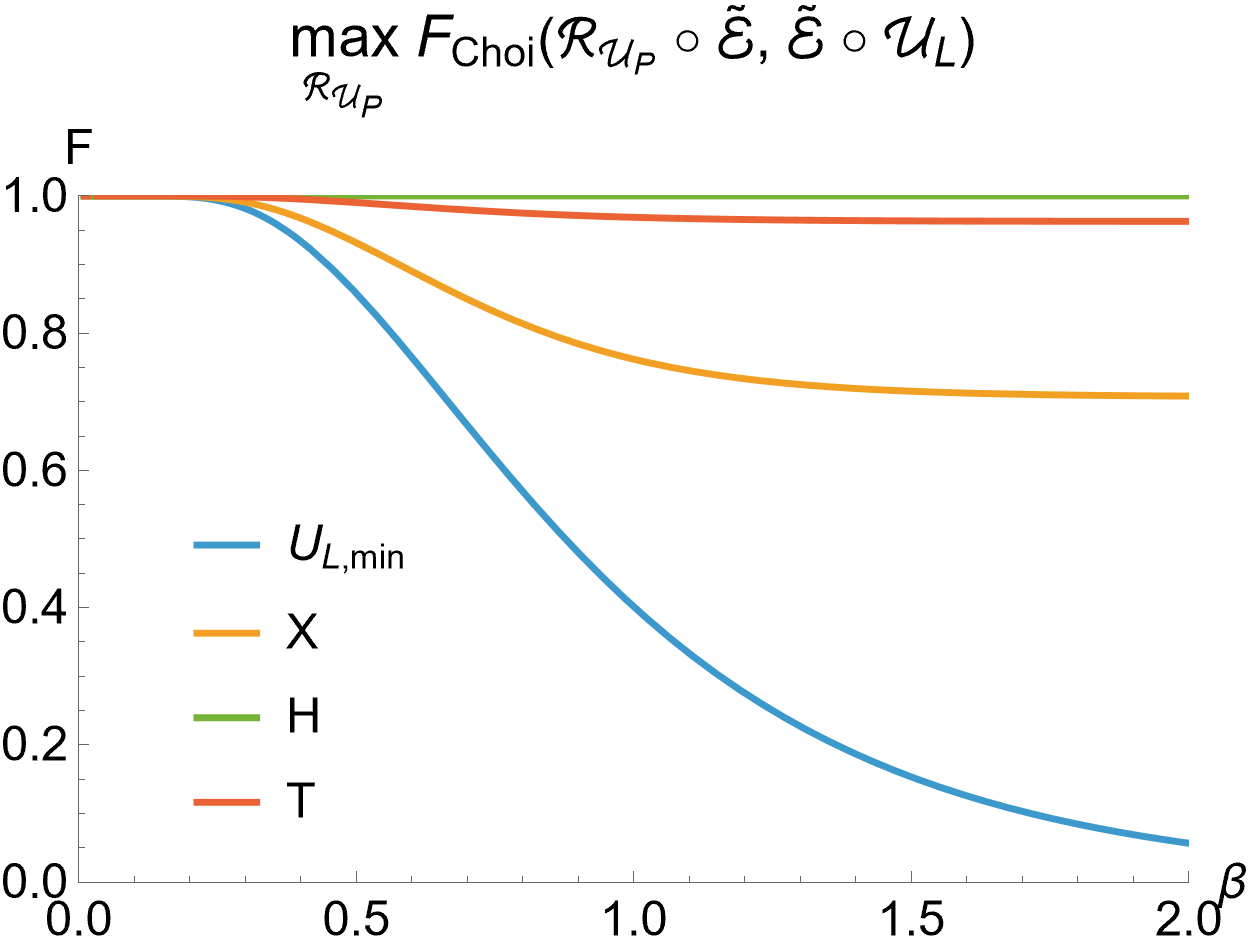}
    \caption{Recovery and logical gate accuracies of regularized GKP codes. Upper panel: $F_{\mathrm{min}}$, $F_{\mathrm{avg}}$ and $F_{\mathrm{Choi}}$ with respect to the regularization parameter $\beta$. Lower panel: $\max_{\mathcal{R}_{\mathcal{U}_P}}F_{\mathrm{Choi}}(\mathcal{R}_{\mathcal{U}_P}\circ\tilde{\mathcal{E}},\tilde{\mathcal{E}}\circ\mathcal{U}_L)$ for logical gates $X$, $H$, $T$ and $U_{L,\mathrm{min}}$ with respect to $\beta$.}
    \label{fig:Fmac}
\end{figure}

\subsection{Tiger codes}
Tiger codes provide a framework for constructing multi-mode bosonic codes without code concatenation~\cite{xu2025letting}. The logical information can be encoded in qubits, qudits, oscillators, or even rotors. Analogous to CSS codes, a tiger code can be specified by check matrices $H_X \in \mathbb{N}^{r_X\times N}$ and $H_Z \in \mathbb{Z}^{r_Z\times N}$ satisfying the CSS condition $H_X H_Z^\T=0$. The matrix $H_X$ has $r_X$ rows; its $i$-th row, denoted by $\mathbf{g}_i$, defines a dissipator. Similarly, $H_Z$ contains $r_Z$ rows; its $i$-th row, denoted by $\mathbf{h}_i$, defines a stabilizer. The dissipators and stabilizers associated with row vectors $\mathbf{g}_i, \mathbf{h}_i$ are given by 
\begin{eqs}
    &\mathbf{g}_i \mapsto \hat{a}_1^{(\mathbf{g}_i)_1}\hat{a}_2^{(\mathbf{g}_i)_2}\cdots\hat{a}_N^{(\mathbf{g}_i)_N},\\
    &\mathbf{h}_i \mapsto e^{i \varphi (\mathbf{h}_i \cdot \hat{\mathbf{n}}-\Delta_i)},~~~ \forall \varphi \in [0,2\pi),\\
\end{eqs}
where $\hat{a}_j$ denotes the annihilation operator acting on the $j$-th mode, and $\hat{\mathbf{n}}=(\hat{n}_1, \hat{n}_2,\ldots, \hat{n}_N)$ is a vector formed by occupation number operators.
Here $\mathbf{\Delta} \in \mathbb{Z}^{r_Z}$ is a constant vector independent of $H_X$ and $H_Z$, and $\Delta_i$ denotes its $i$-th entry. The logical $X$ and $Z$ operators can be determined as the generators of the homology and the cohomology groups defined by $H_X $ and $H_Z$, respectively.

The codewords of the infinite-support tiger codes are coherent states projected onto a eigenspace satisfying $H_Z \cdot \hat{\mathbf{n}}=\mathbf{\Delta}$. We can take coherent states on $N$ modes $\ket{\alpha}^{\otimes N}$ and then project onto the common eigenvalue $1$ eigenbasis of the stabilizers. Since finite-energy coherent states are generally nonorthogonal, the projected states remain nonorthogonal in general. Consequently, finite-energy tiger codes realize non-isometric encodings. In the case that the logical degree of freedom is a qudit, the X-basis codewords can be written as
\begin{eqs}
 \ket{\mu_p}=\frac{e^{N \alpha^2/2}}{\sqrt{A_{H_Z,\mathbf{\Delta}}(\alpha^2,\ldots,\alpha^2)}} \Pi_{H_Z,\mathbf{\Delta}}\bigotimes_{i=1}^N\ket{\alpha e^{-i \frac{2\pi}{d}\mu z_i}},
\end{eqs}
where
\begin{align}
   &A_{H_Z,\mathbf{\Delta}}(\alpha^2,\ldots,\alpha^2)\equiv \sum_{H_Z \mathbf{n}=\mathbf{\Delta}} \prod_{i=1}^N \frac{\alpha^{2 n_i}}{n_i!},\\
   &\Pi_{H_Z,\mathbf{\Delta}} \equiv \sum_{H_Z \mathbf{n}=\mathbf{\Delta}} \ket{\mathbf{n}}\bra{\mathbf{n}},
\end{align}
where $\sum_{H_Z \mathbf{n}=\mathbf{\Delta}}$ denotes a summation over all possible $\mathbf{n} \in \mathbb{N}^N$ satisfying $H_Z \mathbf{n}=\mathbf{\Delta}$, $\ket{\mathbf{n}}\equiv \ket{n_1,\ldots,n_N}$ is a multi-mode Fock state specified by vector $\mathbf{n}$, and $z_i$ denotes the $i$-th entry of vector $z$ representing $\overline{Z}$ operator.

With the explicit form of codewords, we can calculate the eigenvalues and eigenbasis of $V^\dagger V$
\begin{align}
    &(V^\dagger V)_{\nu \mu}=\frac{ A_{H_Z,\mathbf{\Delta}}(\alpha^2 e^{-i \frac{2\pi}{d}(\mu-\nu) z_1},\ldots,\alpha^2 e^{-i \frac{2\pi}{d}(\mu-\nu) z_N})}{ A_{H_Z,\mathbf{\Delta}}(\alpha^2,\ldots,\alpha^2)};\\
    & \lambda_l =\frac{d}{A_{H_Z,\mathbf{\Delta}}(\alpha^2,\ldots,\alpha^2)} \sum_{\substack{H_Z \mathbf{n}=\mathbf{\Delta},\\
    z\cdot\mathbf{n}~\text{mod}~{d}=l}} \prod_i^N \frac{\alpha^{2 n_i}}{n_i!}, ~l\in\mathbb{Z}_d;\\ 
    &\ket{v_l}=\frac{1}{\sqrt{d}}\sum_{\mu=0}^{d-1} e^{i \frac{2\pi}{d} l \mu} \ket{\overline{\mu}}.
\end{align}
The eigenbases of $V^\dagger V$ are identical to the $Z$-basis eigenstates $\ket{\overline{l}}$, $\forall l \in \mathbb{Z}_d$, but their codewords have different normalization.

To give a concrete example, we consider a four-mode tiger code, defined by 
\begin{eqs}
    H_X=\begin{pmatrix}
        1&1&0&0\\
        0&0&1&1\\
        0&2&0&2
    \end{pmatrix},~ H_Z=\begin{pmatrix}
        1&-1&-1&1
    \end{pmatrix}.
\end{eqs}
This code encodes a logical qubit. The logical operators are $x=\begin{pmatrix}
    0&1&0&1
\end{pmatrix}$ and $z=\begin{pmatrix}
    0&0&1&1
\end{pmatrix}$, corresponding to $\overline X=\frac{\hat{a}_2 \hat{a}_4}{\alpha^2}$, $\overline Z=(-1)^{\hat {n}_3+\hat{n}_4}$. Here we consider the $\Delta=0$ case, whose codewords and the encoding map are
\begin{align}
    & \ket{\pm_p}=\frac{e^{2\alpha^2}}{\sqrt{I_0(4\alpha^2)}} \Pi_{H_Z,0}\ket{\alpha,\alpha, \pm \alpha, \pm\alpha},\\
     &V=\ket{+_p}\bra{\overline{+}}+\ket{-_p}\bra{\overline{-}}.
\end{align}
Therefore, we can calculate 
\begin{align}
   &V^\dagger V=\begin{pmatrix}
        1& I_0(4\alpha^2)^{-1}\\
        I_0(4\alpha^2)^{-1}&1
    \end{pmatrix}\\
   & \lambda_0 =1+I_0(4\alpha^2)^{-1},~ \lambda_1 =1-I_0(4\alpha^2)^{-1}.
\end{align}
The eigenbasis of $V^\dagger V$ are the usual computational basis $\ket{\overline{0}}$ and $\ket{\overline{1}}$. Therefore, following Proposition \ref{prop:bestworstUL}, all the phase gates, such as $\overline Z$, $\overline S$, $\overline T$, can be represented exactly by a unitary on the physical Hilbert space. Meanwhile, the logical $\overline X$ gate is the corresponding $U_{L, \mathrm{min}}$.

\subsection{Holographic duality}
Beyond quantum coding theory, non-isometric encoding arises in holographic quantum gravity for fundamental reasons. A representative setting where this mechanism becomes analytically transparent is given by the random matrix model studied in Ref.~\cite{kar2023non}. Following the setup of Ref.~\cite{penington2022replica}, we consider two-dimensional Jackiw--Teitelboim gravity with end-of-world brane states in a microcanonical ensemble associated with an energy window $(E,E+\Delta E)$. In a suitable choice of bases, the bulk-to-boundary map can be realized by a matrix $V$ whose entries are proportional to independent complex Gaussian random variables with zero mean and unit variance: 
\begin{eqs}
    V: \mathcal{H}_{\text{bulk}}\to \mathcal{H}_{\text{boundary}},~V=\frac{1}{\sqrt{d_P}} C,~ C_{ij}\sim\mathcal{CN}(0,1).
\end{eqs}
We identify $\mathcal{H}_{\text{bulk}},\mathcal{H}_{\text{boundary}}$ with the logical and physical Hilbert spaces, respectively.
Say $\dim(\mathcal{H}_{\text{bulk}})=d_L$ and $\dim(\mathcal{H}_{\text{boundary}})=d_P$, which characterizes the number of brane flavors in the gravity theory and the number of boundary states in the aforementioned energy window, respectively. Here, $V^\dagger V$ is a Wishart random matrix acting on the logical Hilbert space. When $d_L>d_P$, it has rank at most $d_P$, and hence has $d_L-d_P$ exact zero eigenvalues. We are interested in the regime where $d_L$ and $d_P$ are both large with $d_L\gg d_P$. In this limit, the eigenvalue distribution of $V^\dagger V$ converges to the Marchenko–Pastur distribution
\begin{align}
    &\rho(\lambda)=\frac{\sqrt{(\lambda_+-\lambda)(\lambda-\lambda_-)}}{2\pi\lambda c}+ \left(1-\frac{1}{c}\right)\delta(\lambda),\\
    &c=\frac{d_L}{d_P},~\lambda_{\pm}=(1\pm \sqrt{c})^2.
\end{align}

With the distribution of the eigenvalues, we see that $F_{\mathrm{min}}=0$, since the encoding map has a kernel when $d_L>d_P$. Then $F_{\mathrm{Choi}}$ can be evaluated as an integral
\begin{eqs}
   F_{\mathrm{Choi}}= \frac{\frac{1}{d_L}\sum_{i=0}^{d_L-1}\lambda^{\frac{1}{2}}_i}{ (\frac{1}{d_L}\sum_{i=0}^{d_L-1}\lambda_i)^{\frac{1}{2}}}=\frac{\int \sqrt{\lambda}\rho(\lambda)d\lambda}{(\int \lambda \rho(\lambda)d\lambda)^{\frac{1}{2}}}.
\end{eqs}
The result is a certain combination of elliptic functions, only depending on the parameter $c$. In the limit $d_L\gg d_P$, thus $c\to\infty$, $F_{\mathrm{Choi}}\sim \sqrt{\frac{1}{c}}=\sqrt{\frac{d_P}{d_L}}$. In the large dimension limit, the nonzero singular values of $V$ become sharply concentrated, so the Choi and average fidelities and errors exhibit the same leading behavior. Therefore, we have a simple interpretation of the Choi fidelity result: the average error rate for recovering generic bulk information from the boundary is approximately $1-\frac{d_P}{d_L}$, which is the ratio between the dimension of the null subspace and the total dimension of the logical Hilbert space.

Gaussian random encodings also appear in other scenarios of quantum information, such as the construction of Pauli operators of overlapping qubits \cite{chao2017overlapping}. The above result of $F_{\mathrm{Choi}}$ implies that, although an exponentially large number of approximate Pauli operators can be constructed, the additional information encoded in this way cannot be recovered with high probability. This is consistent with Nayak’s private information retrieval bound \cite{nayak1999optimal}.

\section{Discussion}

In this work, we established a general information-theoretic framework for quantum error correction with non-isometric encodings and discussed their practical and physical relevance through a series of concrete examples. 
Our results provide a principled foundation for understanding and analyzing the effects of non-isometry on QEC, regardless of its physical origin.
To summarize, the central message is that the overlaps and relative normalizations of the physical codewords, encoded in the Gram operator $V^\dagger V$, impose fundamental limits on the achievable accuracy of logical information preservation and processing.

We close with several outlook remarks highlighting broader implications and future directions.
First, non-isometric maps naturally describe postselected non-unitary dynamics, as encountered in various important areas including dissipative quantum trajectories~\cite{dalibard1992wave,plenio1998quantum} and engineered dissipation approaches to quantum information processing~\cite{verstraete2009quantum,pastawski2011quantum, self_correct_GKP}. {Moreover, recent work developed Hamiltonian-engineering-based methods for approximate codeword preparation~\cite{rymarz2021hardware,kolesnikow2024gkp}. This work provides a natural scheme for evaluating the quality and exploring the utility of such physically motivated encoders.} 
Second, in the continuous variable codes  discussed above, the inaccuracy induced by non-isometry is closely tied to the energy of the codewords. Therefore, our framework can help estimate the energy requirements depending on the target accuracy and thus provide quantitative guidance for experiments.
More broadly, our theory opens the door to adapting and designing quantum codes beyond the standard isometric setting. Allowing controlled non-isometry may provide extra flexibility and better compatibility with physical constraints, at the cost of intrinsic approximation errors. We expect that further study along this direction will broaden the possibilities for fault tolerance, especially in connection with specific physical settings.

\section*{Acknowledgments}

YW thanks Miguel Tierz for discussions. YW is supported by startup funding from SIMIS. YX thanks Yujie Zhang and Yilun Li for hosting his visits in Tokyo. ZWL is supported in part by NSFC under Grant No.~12475023, Dushi Program, and a startup funding from YMSC.

\bibliography{nonisometry}

@article{beny2010general,
  title   = {General conditions for approximate quantum error correction and near-optimal recovery channels},
  author  = {B{\'e}ny, C{\'e}dric and Oreshkov, Ognyan},
  journal = {Phys. Rev. Lett.},
  volume  = {104},
  number  = {12},
  pages   = {120501},
  year    = {2010},
  doi     = {10.1103/PhysRevLett.104.120501}
}

@article{antonini2025non,
  title   = {Non-isometry, state dependence and holography},
  author  = {Antonini, Stefano and Balasubramanian, Vijay and Bao, Ning and Cao, ChunJun and Chemissany, Wissam},
  journal = {J. High Energ. Phys.},
  volume  = {2025},
  number  = {2},
  pages   = {150},
  year    = {2025},
  doi     = {10.1007/JHEP02(2025)150}
}

@article{sion1958general,
  title   = {On general minimax theorems},
  author  = {Sion, Maurice},
  journal = {Pacific J. Math.},
  volume  = {8},
  number  = {1},
  pages   = {171--176},
  year    = {1958},
  doi     = {10.2140/pjm.1958.8.171}
}

@article{araki1990inequality,
  title   = {On an inequality of Lieb and Thirring},
  author  = {Araki, Huzihiro},
  journal = {Lett. Math. Phys.},
  volume  = {19},
  number  = {2},
  pages   = {167--170},
  year    = {1990},
  doi     = {10.1007/BF01045887}
}

@article{audenaert2007araki,
  title   = {On the {Araki-Lieb-Thirring} inequality},
  author  = {Audenaert, Koenraad M. R.},
  journal = {Int. J. Inf. Syst. Sci.},
  volume  = {4},
  number  = {1},
  pages   = {78--83},
  year    = {2008}
}

@article{kim2020ghost,
  title   = {The ghost in the radiation: Robust encodings of the black hole interior},
  author  = {Kim, Isaac H. and Tang, Eugene and Preskill, John},
  journal = {J. High Energ. Phys.},
  volume  = {2020},
  number  = {6},
  pages   = {031},
  year    = {2020},
  doi     = {10.1007/JHEP06(2020)031}
}

@article{akers2024black,
  title   = {The black hole interior from non-isometric codes and complexity},
  author  = {Akers, Chris and Engelhardt, Netta and Harlow, Daniel and Penington, Geoff and Vardhan, Shreya},
  journal = {J. High Energ. Phys.},
  volume  = {2024},
  number  = {6},
  pages   = {155},
  year    = {2024},
  doi     = {10.1007/JHEP06(2024)155}
}

@article{kar2023non,
  title   = {Non-isometric quantum error correction in gravity},
  author  = {Kar, Arjun},
  journal = {J. High Energ. Phys.},
  volume  = {2023},
  number  = {2},
  pages   = {195},
  year    = {2023},
  doi     = {10.1007/JHEP02(2023)195}
}

@article{cao2025overlapping,
  title   = {Overlapping qubits from non-isometric maps and de {Sitter} tensor networks},
  author  = {Cao, ChunJun and Chemissany, Wissam and Jahn, Alexander and Zimbor{\'a}s, Zolt{\'a}n},
  journal = {Nat. Commun.},
  volume  = {16},
  number  = {1},
  pages   = {163},
  year    = {2025},
  doi     = {10.1038/s41467-024-55463-9}
}

@article{antonini2023cosmology,
  title   = {Cosmology from random entanglement},
  author  = {Antonini, Stefano and Sasieta, Martin and Swingle, Brian},
  journal = {J. High Energ. Phys.},
  volume  = {2023},
  number  = {11},
  pages   = {188},
  year    = {2023},
  doi     = {10.1007/JHEP11(2023)188}
}

@article{akers2025observers,
  title   = {On observers in holographic maps},
  author  = {Akers, Chris and Bueller, Gracemarie and DeWolfe, Oliver and Higginbotham, Kenneth and Reinking, Johannes and Rodriguez, Rudolph},
  journal = {J. High Energ. Phys.},
  volume  = {2025},
  number  = {5},
  pages   = {201},
  year    = {2025},
  doi     = {10.1007/JHEP05(2025)201}
}

@article{papadodimas2014state,
  title   = {State-dependent bulk-boundary maps and black hole complementarity},
  author  = {Papadodimas, Kyriakos and Raju, Suvrat},
  journal = {Phys. Rev. D},
  volume  = {89},
  number  = {8},
  pages   = {086010},
  year    = {2014},
  doi     = {10.1103/PhysRevD.89.086010}
}

@article{dewolfe2023non,
  title   = {Non-isometric codes for the black hole interior from fundamental and effective dynamics},
  author  = {DeWolfe, Oliver and Higginbotham, Kenneth},
  journal = {J. High Energ. Phys.},
  volume  = {2023},
  number  = {9},
  pages   = {068},
  year    = {2023},
  doi     = {10.1007/JHEP09(2023)068}
}

@article{Kong2022NearOptimal,
  title   = {Near-optimal covariant quantum error-correcting codes from random unitaries with symmetries},
  author  = {Kong, Linghang and Liu, Zi-Wen},
  journal = {PRX Quantum},
  volume  = {3},
  number  = {2},
  pages   = {020314},
  year    = {2022},
  doi     = {10.1103/PRXQuantum.3.020314}
}

@misc{faulkner2022asymptotically,
  title         = {Asymptotically isometric codes for holography},
  author        = {Faulkner, Thomas and Li, Min},
  year          = {2022},
  eprint        = {2211.12439},
  archivePrefix = {arXiv},
  primaryClass  = {hep-th},
  note          = {arXiv:2211.12439}
}

@inproceedings{chao2017overlapping,
  title     = {Overlapping qubits},
  author    = {Chao, Rui and Reichardt, Ben W. and Sutherland, Chris and Vidick, Thomas},
  booktitle = {8th Innovations in Theoretical Computer Science Conference ({ITCS} 2017)},
  series    = {Leibniz International Proceedings in Informatics ({LIPIcs})},
  volume    = {67},
  pages     = {48:1--48:21},
  year      = {2017},
  doi       = {10.4230/LIPIcs.ITCS.2017.48}
}

@inproceedings{nayak1999optimal,
  title     = {Optimal lower bounds for quantum automata and random access codes},
  author    = {Nayak, Ashwin},
  booktitle = {Proceedings of the 40th Annual Symposium on Foundations of Computer Science},
  pages     = {369--376},
  year      = {1999},
  doi       = {10.1109/SFFCS.1999.814608}
}

@article{gottesman2000encoding,
  title   = {Encoding a qubit in an oscillator},
  author  = {Gottesman, Daniel and Kitaev, Alexei and Preskill, John},
  journal = {Phys. Rev. A},
  volume  = {64},
  number  = {1},
  pages   = {012310},
  year    = {2001},
  doi     = {10.1103/PhysRevA.64.012310}
}

@article{jafarzadeh2025logical,
  title   = {Logical channels in approximate {Gottesman-Kitaev-Preskill} error correction},
  author  = {Jafarzadeh, Mahnaz and Conrad, Jonathan and Alexander, Rafael N. and Baragiola, Ben Q.},
  journal = {Phys. Rev. A},
  volume  = {112},
  number  = {6},
  pages   = {062413},
  year    = {2025},
  doi     = {10.1103/c8hk-v1qf}
}

@article{xu2025letting,
  title   = {Letting the tiger out of its cage: Bosonic coding without concatenation},
  author  = {Xu, Yijia and Wang, Yixu and Vuillot, Christophe and Albert, Victor V.},
  journal = {Phys. Rev. X},
  volume  = {15},
  number  = {4},
  pages   = {041025},
  year    = {2025},
  doi     = {10.1103/ls5r-vj7r}
}

@article{tzitrin2020progress,
  title   = {Progress towards practical qubit computation using approximate {Gottesman-Kitaev-Preskill} codes},
  author  = {Tzitrin, Ilan and Bourassa, J. Eli and Menicucci, Nicolas C. and Sabapathy, Krishna Kumar},
  journal = {Phys. Rev. A},
  volume  = {101},
  number  = {3},
  pages   = {032315},
  year    = {2020},
  doi     = {10.1103/PhysRevA.101.032315}
}

@article{eaton2019non,
  title   = {Non-Gaussian and {Gottesman-Kitaev-Preskill} state preparation by photon catalysis},
  author  = {Eaton, Miller and Nehra, Rajveer and Pfister, Olivier},
  journal = {New J. Phys.},
  volume  = {21},
  number  = {11},
  pages   = {113034},
  year    = {2019},
  doi     = {10.1088/1367-2630/ab5330}
}

@article{penington2022replica,
  title   = {Replica wormholes and the black hole interior},
  author  = {Penington, Geoff and Shenker, Stephen H. and Stanford, Douglas and Yang, Zhenbin},
  journal = {J. High Energ. Phys.},
  volume  = {2022},
  number  = {3},
  pages   = {205},
  year    = {2022},
  doi     = {10.1007/JHEP03(2022)205}
}

@article{gilchrist2005distance,
  title   = {Distance measures to compare real and ideal quantum processes},
  author  = {Gilchrist, Alexei and Langford, Nathan K. and Nielsen, Michael A.},
  journal = {Phys. Rev. A},
  volume  = {71},
  number  = {6},
  pages   = {062310},
  year    = {2005},
  doi     = {10.1103/PhysRevA.71.062310}
}

@article{horodecki1999general,
  title   = {General teleportation channel, singlet fraction, and quasidistillation},
  author  = {Horodecki, Micha{\l} and Horodecki, Pawe{\l} and Horodecki, Ryszard},
  journal = {Phys. Rev. A},
  volume  = {60},
  number  = {3},
  pages   = {1888--1898},
  year    = {1999},
  doi     = {10.1103/PhysRevA.60.1888}
}

@article{mirrahimi2014dynamically,
  title={Dynamically protected cat-qubits: a new paradigm for universal quantum computation},
  author={Mirrahimi, Mazyar and Leghtas, Zaki and Albert, Victor V and Touzard, Steven and Schoelkopf, Robert J and Jiang, Liang and Devoret, Michel H},
  journal={New Journal of Physics},
  volume={16},
  number={4},
  pages={045014},
  year={2014},
  publisher={IOP Publishing},
  url={https://doi.org/10.1088/1367-2630/16/4/045014},
  doi={10.1088/1367-2630/16/4/045014}
}

@article{menicucci2014fault,
  title = {Fault-Tolerant Measurement-Based Quantum Computing with Continuous-Variable Cluster States},
  author = {Menicucci, Nicolas C.},
  journal = {Phys. Rev. Lett.},
  volume = {112},
  issue = {12},
  pages = {120504},
  numpages = {5},
  year = {2014},
  month = {Mar},
  publisher = {American Physical Society},
  doi = {10.1103/PhysRevLett.112.120504},
  url = {https://link.aps.org/doi/10.1103/PhysRevLett.112.120504}
}

@article{kl1997,
  title = {Theory of quantum error-correcting codes},
  author = {Knill, Emanuel and Laflamme, Raymond},
  journal = {Phys. Rev. A},
  volume = {55},
  issue = {2},
  pages = {900--911},
  numpages = {0},
  year = {1997},
  month = {Feb},
  publisher = {American Physical Society},
  doi = {10.1103/PhysRevA.55.900},
  url = {https://link.aps.org/doi/10.1103/PhysRevA.55.900}
}

@article{brenner2025composable,
  title={Composable logical gate error in approximate quantum error correction: reexamining gate implementations in Gottesman-Kitaev-Preskill codes},
  author={Brenner, Lukas and Dias, Beatriz and Koenig, Robert},
  journal={arXiv preprint arXiv:2509.14658},
  year={2025}
}

@article{dalibard1992wave,
title = {Wave-function approach to dissipative processes in quantum optics},
author = {Dalibard, J. and Castin, Y. and M{\o}lmer, K.},
journal = {Phys. Rev. Lett.},
volume = {68},
pages = {580--583},
year = {1992},
doi = {10.1103/PhysRevLett.68.580}
}

@article{plenio1998quantum,
title = {The quantum-jump approach to dissipative dynamics in quantum optics},
author = {Plenio, M. B. and Knight, P. L.},
journal = {Rev. Mod. Phys.},
volume = {70},
pages = {101--144},
year = {1998},
doi = {10.1103/RevModPhys.70.101}
}

@article{verstraete2009quantum,
title = {Quantum computation and quantum-state engineering driven by dissipation},
author = {Verstraete, F. and Wolf, M. M. and Cirac, J. I.},
journal = {Nat. Phys.},
volume = {5},
pages = {633--636},
year = {2009},
doi = {10.1038/nphys1342}
}

@article{pastawski2011quantum,
title = {Quantum memories based on engineered dissipation},
author = {Pastawski, F. and Clemente, L. and Cirac, J. I.},
journal = {Phys. Rev. A},
volume = {83},
pages = {012304},
year = {2011},
doi = {10.1103/PhysRevA.83.012304}
}

@article{self_correct_GKP,
  title = {Self-Correcting Gottesman-Kitaev-Preskill Qubit and Gates in a Driven-Dissipative Circuit},
  author = {Nathan, Frederik and O'Brien, Liam and Noh, Kyungjoo and Matheny, Matthew H. and Grimsmo, Arne L. and Jiang, Liang and Refael, Gil},
  journal = {PRX Quantum},
  volume = {6},
  issue = {3},
  pages = {030352},
  numpages = {31},
  year = {2025},
  month = {Sep},
  publisher = {American Physical Society},
  doi = {10.1103/ykqb-m52z},
  url = {https://link.aps.org/doi/10.1103/ykqb-m52z}
}

@article{rymarz2021hardware,
  title = {Hardware-Encoding Grid States in a Nonreciprocal Superconducting Circuit},
  author = {Rymarz, Martin and Bosco, Stefano and Ciani, Alessandro and DiVincenzo, David P.},
  journal = {Phys. Rev. X},
  volume = {11},
  issue = {1},
  pages = {011032},
  numpages = {25},
  year = {2021},
  month = {Feb},
  publisher = {American Physical Society},
  doi = {10.1103/PhysRevX.11.011032},
  url = {https://link.aps.org/doi/10.1103/PhysRevX.11.011032}
}

@article{kolesnikow2024gkp,
  title = {Gottesman-Kitaev-Preskill State Preparation Using Periodic Driving},
  author = {Kolesnikow, Xanda C. and Bomantara, Raditya W. and Doherty, Andrew C. and Grimsmo, Arne L.},
  journal = {Phys. Rev. Lett.},
  volume = {132},
  issue = {13},
  pages = {130605},
  numpages = {6},
  year = {2024},
  month = {Mar},
  publisher = {American Physical Society},
  doi = {10.1103/PhysRevLett.132.130605},
  url = {https://link.aps.org/doi/10.1103/PhysRevLett.132.130605}
}

\newpage

\onecolumngrid

\appendix
\setcounter{theorem}{0}
\renewcommand{\thetheorem}{S\arabic{theorem}}

\section{An information-disturbance theorem for non-isometric encodings} \label{sec:theoremproof}

We formulate the following result, which relates the optimal recovered fidelity of a pair of maps to the corresponding fidelity between their complementary maps. If a general map can be written in the form $\mathcal{N}(\rho)=\sum_{i}E^{(\mathcal{N})}_i \rho E_{i}^{(\mathcal{N})\dagger}$, then its complementary map can be defined as $\mathcal{N}_c(\rho)=\sum_{i,j}\tr\left(E_i^{(\mathcal{N})}\rho E_{j}^{(\mathcal{N})\dagger}\right)\ket{i}\bra{j}$. This map can either be a quantum channel or a non-isometric map, such as $\tilde{\mathcal{E}}$, or their composition.

This result is an extension of the theorem for channels in Ref.~\cite{beny2010general}. In  the special case $(n_1, n_2)=(0,0)$, the above result reduces to the one in Ref.~\cite{beny2010general}. However, since $\tilde{\mathcal{E}}$ is not a quantum channel, the remaining cases of the present theorem require a separate proof. 

\begin{theorem} \label{theorem:fidelitycomplementary}
     Let $\tilde{\mathcal{E}}$ be defined as in (\ref{eq:nonisoencoding}), where $V=\sum_{i=0}^{d-1} \ket{i_p}\bra{\overline i}$. Then for fixed quantum channels $\mathcal{N}_1$, $\mathcal{N}_2$, $\mathcal{M}_1$ and $\mathcal{M}_2$, together with a pair of binary numbers $(n_1,n_2)\in \mathbb{Z}_2^2$, 
\begin{align}
    &\max_{\mathcal{R}}F_{\rho}(\mathcal{R}\circ \mathcal{N}_2\circ (\tilde{\mathcal{E}})^{n_1}\circ \mathcal{N}_1,\mathcal{M}_2\circ (\tilde{\mathcal{E}})^{n_2}\circ \mathcal{M}_1)=\max_{\mathcal{R^{\prime}}}F_{\rho}((\mathcal{N}_2 \circ (\tilde{\mathcal{E}})^{n_1} \circ \mathcal{N}_1)_c,\mathcal{R^{\prime}}\circ (\mathcal{M}_2\circ (\tilde{\mathcal{E}})^{n_2}\circ \mathcal{M}_1)_c),~\forall \rho \in \mathcal{D}(\mathcal{H}_L),
    \end{align}
where $\mathcal{D}(\mathcal{H}_L)$ denotes the space of density operators on $\mathcal{H}_L$, and
    \begin{align}
    &\max_{\mathcal{R}}F(\mathcal{R}\circ \mathcal{N}_2\circ (\tilde{\mathcal{E}})^{n_1}\circ \mathcal{N}_1,\mathcal{M}_2\circ (\tilde{\mathcal{E}})^{n_2}\circ \mathcal{M}_1)=\max_{\mathcal{R^{\prime}}}F((\mathcal{N}_2 \circ (\tilde{\mathcal{E}})^{n_1} \circ \mathcal{N}_1)_c,\mathcal{R^{\prime}}\circ (\mathcal{M}_2\circ (\tilde{\mathcal{E}})^{n_2}\circ \mathcal{M}_1)_c).\\
\end{align}
Here $n_i\in\{0,1\}$ for $i=1,2$ indicates whether the non-isometric encoding is inserted: $(\tilde{\mathcal E})^1=\tilde{\mathcal E}$ and $(\tilde{\mathcal E})^0=\mathrm{id}$.
The other quantum channels are chosen such that both $\mathcal{R}\circ \mathcal{N}_2\circ (\tilde{\mathcal{E}})^{n_1}\circ \mathcal{N}_1$ and $\mathcal{M}_2\circ (\tilde{\mathcal{E}})^{n_2}\circ \mathcal{M}_1$ have the same target Hilbert space.

\end{theorem}

\begin{proof}

Denote the system of interest as $S$ (the logical space $\mathcal{H}_L$ in the main text) and its purifier as $P$. A generic density matrix $\rho_S$ has its purified state $\ket{\psi_{\rho}}_{SP}$. For a generic quantum channel $\mathcal{N}$ acting on $S$ represented as $\mathcal{N}(\rho_S)=\sum_{i}E_i^{(\mathcal{N})}\rho_SE^{(\mathcal{N})\dagger}_i$, we can equivalently define an isometry $V^{(\mathcal{N})}_{SE}=\sum_{i}E_i^{(\mathcal{N})}\otimes\ket{i}_{E}$ acting on the system and the environment $E$, such that $\mathcal{N}(\rho_S)=\tr_{E} V^{(\mathcal{N})}_{SE}\rho_SV^{(\mathcal{N})\dagger}_{SE}$ and $\mathcal{N}_c(\rho_S)=\tr_{S} V^{(\mathcal{N})}_{SE}\rho_S V^{(\mathcal{N})\dagger}_{SE}=\sum_{i,j}\tr(E_i^{(\mathcal{N})}\rho_S E_{j}^{(\mathcal{N})\dagger})\ket{i}\bra{j}_E$. Equivalently, the Stinespring dilation theorem defines a unitary $U_{\mathcal{N}}:\ket{\psi_{\rho}}_S\otimes\ket{0}_{E}\mapsto\sum_{i}E_i^{(\mathcal{N})}\ket{\psi_{\rho}}_{S}\otimes\ket{i}_{E}$, such that  $\mathcal{N}(\rho_S)=\tr_{E} U^{(\mathcal{N})}_{SE}\rho_S\otimes\ket{0}\bra{0}_E U^{(\mathcal{N})\dagger}_{SE}$.

Within this framework, we denote the composite maps $\mathcal{N}_2\circ (\tilde{\mathcal{E}})^{n_1}\circ \mathcal{N}_1=\mathcal{\overline{N}}$,~$\mathcal{M}_2\circ (\tilde{\mathcal{E}})^{n_2}\circ \mathcal{M}_1=\mathcal{\overline{M}}$ and write 
    \begin{align}
    \mathcal{\overline{N}}\otimes \mathrm{id} (\ket{\psi_{\rho}}\bra{\psi_{\rho}}_{SP})&=\frac{\tr_{E} V^{(\mathcal{\overline{N}})}_{SE}(\ket{\psi_{\rho}}\bra{\psi_{\rho}}_{SP})V^{(\mathcal{\overline{N}})\dagger}_{SE}}{\tr_{SPE} V^{(\mathcal{\overline{N}})}_{SE}(\ket{\psi_{\rho}}\bra{\psi_{\rho}}_{SP})V^{(\mathcal{\overline{N}})\dagger}_{SE}},\\
    \mathcal{\overline{N}}_c\otimes \mathrm{id} (\ket{\psi_{\rho}}\bra{\psi_{\rho}}_{SP})&=\frac{\tr_{S} V^{(\mathcal{\overline{N}})}_{SE}(\ket{\psi_{\rho}}\bra{\psi_{\rho}}_{SP})V^{(\mathcal{\overline{N}})\dagger}_{SE}}{\tr_{SPE}V^{(\mathcal{\overline{N}})}_{SE}(\ket{\psi_{\rho}}\bra{\psi_{\rho}}_{SP})V^{(\mathcal{\overline{N}})\dagger}_{SE}},\\
    \mathcal{\overline{M}}\otimes \mathrm{id} (\ket{\psi_{\rho}}\bra{\psi_{\rho}}_{SP})&=\frac{\tr_{E} V^{(\mathcal{\overline{M}})}_{SE}(\ket{\psi_{\rho}}\bra{\psi_{\rho}}_{SP})V^{(\mathcal{\overline{M}})\dagger}_{SE}}{\tr_{SPE} V^{(\mathcal{\overline{M}})}_{SE}(\ket{\psi_{\rho}}\bra{\psi_{\rho}}_{SP})V^{(\mathcal{\overline{M}})\dagger}_{SE}},\\
    \mathcal{\overline{M}}_c\otimes \mathrm{id} (\ket{\psi_{\rho}}\bra{\psi_{\rho}}_{SP})&=\frac{\tr_{S} V^{(\mathcal{\overline{M}})}_{SE}(\ket{\psi_{\rho}}\bra{\psi_{\rho}}_{SP})V^{(\mathcal{\overline{M}})\dagger}_{SE}}{\tr_{SPE}V^{(\mathcal{\overline{M}})}_{SE}(\ket{\psi_{\rho}}\bra{\psi_{\rho}}_{SP})V^{(\mathcal{\overline{M}})\dagger}_{SE}},
\end{align}
in which $V^{(\mathcal{\overline{N}})}_{SE}=\sum_{ij} E_i^{(\mathcal{N}_2)}V^{n_1}E_j^{(\mathcal{N}_1)}\otimes \ket{ij}_E$ with $V=\sum_{i=0}^{d-1} \ket{i_p}\bra{\overline i}$, and $V^{(\mathcal{\overline{M}})}_{SE}=\sum_{kl} E_k^{(\mathcal{M}_2)}V^{n_2}E_l^{(\mathcal{M}_1)}\otimes \ket{kl}_E$. We abbreviate  $\tr_{SPE}V^{(\mathcal{\overline{N}})}_{SE}(\ket{\psi_{\rho}}\bra{\psi_{\rho}}_{SP})V^{(\mathcal{\overline{N}})\dagger}_{SE}=n(\rho)$ and $\tr_{SPE}V^{(\mathcal{\overline{M}})}_{SE}(\ket{\psi_{\rho}}\bra{\psi_{\rho}}_{SP})V^{(\mathcal{\overline{M}})\dagger}_{SE}=m(\rho)$. Here $E$ is the environment that is large enough for both $\mathcal{\overline N}$ and $\mathcal{\overline M}$ to act on. The purified state of $\mathcal{\overline{M}}\otimes \mathrm{id} (\ket{\psi_{\rho}}\bra{\psi_{\rho}}_{SP})$ can be taken as $\frac{1}{\sqrt{m(\rho)}}V^{(\mathcal{\overline M})}_{SE}\ket{\psi_{\rho}}_{SP}\ket{0}_{E_R}$. 

Adding the recovery channel, we have 
\begin{eqs}
&\left(\mathcal{R}\circ\mathcal{\overline{N}}\right)\otimes \mathrm{id} (\ket{\psi_{\rho}}\bra{\psi_{\rho}}_{SP})=\tr_{E_R E} \left(U^{(\mathcal{R})}_{SE_R}V^{(\mathcal{\overline{N}})}_{SE}\left(\frac{\ket{\psi_{\rho}}\bra{\psi_{\rho}}_{SP}}{n(\rho)}\otimes\ket{0}\bra{0}_{E_R}\right)V^{(\mathcal{\overline{N}})\dagger}_{SE}U^{(\mathcal{R})\dagger}_{SE_R}\right),
\end{eqs}
with purified state $\frac{1}{\sqrt{n(\rho)}}U^{(\mathcal{R})}_{SE_R}V^{(\mathcal{\overline{N}})}_{SE}\ket{\psi_{\rho}}_{SP}\ket{0}_{E_R}$.

On the other hand,
\begin{eqs}
&\left(\mathcal{R}^{\prime}\circ\mathcal{\overline M}_c\right)\otimes \mathrm{id} (\ket{\psi_{\rho}}\bra{\psi_{\rho}}_{SP})=\tr_{SE_{R^\prime}} \left(U^{(\mathcal{R}^\prime)}_{EE_{R^\prime}}V^{(\mathcal{\overline M})}_{SE}\left(\frac{\ket{\psi_{\rho}}\bra{\psi_{\rho}}_{SP}}{m(\rho)}\otimes\ket{0}\bra{0}_{E_{R^\prime}}\right)V^{(\mathcal{\overline M})\dagger}_{SE}U^{(\mathcal{R}^\prime)\dagger}_{EE_{R^\prime}}\right),
\end{eqs}
with purified state taken as $\frac{1}{\sqrt{m(\rho)}}U^{(\mathcal{R}^\prime)}_{EE_R^\prime}V^{(\mathcal{\overline M})}_{SE}\ket{\psi_{\rho}}\ket{0}_{E_{R^{\prime}}}$. The purified state of $\mathcal{\overline N}_c\otimes \mathrm{id} (\ket{\psi_{\rho}}\bra{\psi_{\rho}}_{SP})$ can be taken as $\frac{1}{\sqrt{n(\rho)}}V^{(\mathcal{\overline N})}_{SE}\ket{\psi_{\rho}}_{SP}\ket{0}_{E_{R^\prime}}$.

Using Uhlmann's representation of fidelity, we can write
    \begin{align}
    &\max_{\mathcal{R}}F_{\rho}(\mathcal{R}\circ \mathcal{\overline N},\mathcal{\overline M})=\max_{U^{(\mathcal{R})}_{SE_R}}\max_{U_{EE_R}}\frac{1}{\sqrt{n(\rho)m(\rho)}}|\bra{0}_{E_R}\bra{\psi_{\rho}}_{SP}V^{(\mathcal{\overline M})\dagger}_{SE} U_{EE_R} U^{(\mathcal{R})}_{SE_R}V^{(\mathcal{\overline N})}_{SE}\ket{\psi_{\rho}}_{SP}\ket{0}_{E_R}|,\\
    &\max_{\mathcal{R^{\prime}}}F_{\rho}(\mathcal{\overline N}_c,\mathcal{R^{\prime}}\circ \mathcal{\overline M}_c)=\max_{U^{(\mathcal{R}^\prime)\dagger}_{EE_{R^\prime}}}\max_{U_{SE_{R^\prime}}}\frac{1}{\sqrt{n(\rho)m(\rho)}}|\bra{0}_{E_{R^\prime}}\bra{\psi_{\rho}}_{SP}V^{(\mathcal{\overline M})\dagger}_{SE}U^{(\mathcal{R}^\prime)\dagger}_{EE_R^\prime}U_{SE_{R^\prime}}V^{(\mathcal{\overline N})}_{SE}\ket{\psi_{\rho}}_{SP}\ket{0}_{E_{R^\prime}}|.
    \end{align}
We see that these two quantities are indeed equivalent by identifying the labels $E_R$ with $E_{R^\prime}$ and exchanging the maximizations over two unitaries.

We denote the objective function of optimization as $g(\rho, U_1, U_2)\equiv\frac{1}{\sqrt{n(\rho)m(\rho)}}\bra{0}_{E_R}\bra{\psi_{\rho}}_{SP}V^{(\mathcal{\overline M})\dagger}_{SE} U_{EE_R} U^{(\mathcal{R})}_{SE_R}V^{(\mathcal{\overline N})}_{SE}\ket{\psi_{\rho}}_{SP}\ket{0}_{E_R},$ where $U_1\equiv U_{SE_R}$ and $U_2\equiv U_{EE_R}$. Now
\begin{align}\label{eq:proofmaxfidelity}
    &\max_{\mathcal{R}}F(\mathcal{R}\circ\mathcal{\overline N},\mathcal{\overline M})=\max_{U_1}\min_{\rho}\max_{U_2}|g(\rho, U_1, U_2)|,\\
    &\max_{\mathcal{R^{\prime}}}F(\mathcal{\overline N}_c,\mathcal{R^{\prime}}\circ \mathcal{\overline M}_c)=\max_{U_2}\min_{\rho}\max_{U_1}|g(\rho, U_1, U_2)|. \label{eq:proofmaxfidelity2}
\end{align}
Similar to the proof in Ref.~\cite{beny2010general}, the rightmost maximization over $|g(\rho, U_1, U_2)|$ is equivalent to the maximization over $\mathrm{Re}\,g(\rho, U_1, U_2)$. The maximizations over $U_1$ and $U_2$ are equivalent to the maximization over the convex set of operators $\|A_1\|\leq1$ and $\|A_2\|\leq1$. For fixed $\rho$, $\mathrm{Re}\,g(\rho, A_1, A_2)$ is linear in $A_1$ and $A_2$. It depends on the density matrix in the form $\tr(\rho \hat{O})/\sqrt{n(\rho)m(\rho)}$, with $\hat{O}$ an operator depending on $V^{(\mathcal{\overline N})}_{SE}$, $V^{(\mathcal{\overline M})}_{SE}$ and $A_1$, $A_2$. The normalization constants $n(\rho)$ and $m(\rho)$ are both linear in $\rho$ as well. 

We now show that $\text{Re}\,g(\rho, A_1, A_2)$ is quasiconvex with respect to $\rho$. A function $f$ is quasiconvex, if $\forall x_1, x_2$ and $\forall \lambda\in[0,1]$,
\[f(\lambda  x_1+(1-\lambda)x_2)\leq \max\{f(x_1),f(x_2)\}.\] 
For two density matrices $\rho_1$ and $\rho_2$, without loss of generality, assume $\tr(\rho_1\hat{O})/\sqrt{n(\rho_1)m(\rho_1)}\geq\tr(\rho_2\hat{O})/\sqrt{n(\rho_2)m(\rho_2)}$. For ease of notation, denote $c_{1/2}=\tr(\rho_{1/2}\hat{O})$ and $n_{1/2}=n(\rho_{1/2})$, $m_{1/2}=m(\rho_{1/2})$. Indeed,
\begin{align}
    &\frac{\lambda c_1+{(1-\lambda)c_2}}{\sqrt{(\lambda n_1+(1-\lambda)n_2)(\lambda m_1+(1-\lambda)m_2)}}\leq \frac{c_1}{\sqrt{n_1 m_1}}\\
    \Leftrightarrow\quad & n_1 m_1(\lambda^2 c_1^2+(1-\lambda)^2c_2^2+2\lambda(1-\lambda)c_1 c_2) \leq  c_1^2(\lambda n_1 +(1-\lambda)n_2)(\lambda m_1 +(1-\lambda)m_2)\\
    \Leftrightarrow\quad & (1-\lambda)^2 (c_1^2 n_2 m_2-c_2^2 n_1 m_1) +\lambda(1-\lambda)(c_1^2(n_1 m_2+m_1 n_2)-2c_1c_2 n_1 m_1)\geq 0.
\end{align}
Since both sides of the original inequality are positive, we can square both sides to obtain the second inequality. The second line is rearranged to become the third line. We see that 
\begin{align}
(c_1^2 n_2 m_2-c_2^2 n_1 m_1)&\geq0,\\
c_1^2(n_1 m_2+m_1 n_2)-2c_1c_2 n_1 m_1&\geq 2c_1^2\sqrt{n_1m_1n_2m_2}-2c_1c_2 n_1 m_1\notag\\&= 2c_1 \sqrt{n_1 m_1}(c_1\sqrt{n_2 m_2}-c_2\sqrt{n_1 m_1})\geq0,    
\end{align} 
by assumption. Therefore, $\mathrm{Re}\,g(\rho, A_1, A_2)$ is quasiconvex with respect to $\rho$.

Since $\mathrm{Re}\,g(\rho, A_1, A_2)$ is quasiconvex with respect to $\rho$ and linear with respect to $A_1$ and $A_2$, we can apply Sion's minimax theorem \cite{sion1958general}, which says we can exchange the maximization and minimization if the target function is upper semicontinuous and quasi-concave on the argument to be maximized, and lower semicontinuous and quasi-convex on the argument to be minimized. Therefore, we can exchange the orders of the maximizations and the minimization in Eqs.~(\ref{eq:proofmaxfidelity}) and (\ref{eq:proofmaxfidelity2}), so the two quantities are equivalent. 

\end{proof}

\section{Calculations and proofs related to optimal recovery fidelities}\label{sec:propcalculation}

\subsection{Worst-case fidelity}
Applying Theorem~\ref{theorem:fidelitycomplementary} 
for the case $(n_1,n_2)=(1,0)$, with $\mathcal{\overline N}=\tilde{\mathcal{E}}$,  $\mathcal{\overline M}=\mathrm{id}$ and $\mathcal{\overline M}_c=\tr_S$, we obtain
\begin{eqs}
F_{\min}&=\max_{\mathcal{R}}F(\mathcal{R}\circ\tilde{\mathcal{E}},\mathrm{id})=\max_{\mathcal{R^{\prime}}}F(\tilde{\mathcal{E}}_c,\mathcal{R^{\prime}}\circ \tr_S)=\max_{\mathcal{R^{\prime}}}\min_{\rho}F_{\rho}(\tilde{\mathcal{E}}_c,\mathcal{R^{\prime}}\circ \tr_S).
\end{eqs}
The output states are
\begin{align}
&\tilde{\mathcal{E}}_c\otimes \mathrm{id}(\ket{\psi_{\rho}}\bra{\psi_{\rho}}_{SP})=\frac{\tr_{S} (V\ket{\psi_{\rho}}\bra{\psi_{\rho}}_{SP}V^{\dagger})}{\tr(V\rho V^\dagger)}\otimes\ket{0}\bra{0}_E,\\
&(\mathcal{R}^{\prime})\circ\tr_S\otimes \mathrm{id} (\ket{\psi_{\rho}}\bra{\psi_{\rho}}_{SP})=\tr_{S} (\ket{\psi_{\rho}}\bra{\psi_{\rho}}_{SP})\otimes \mathcal{R}^{\prime}(\ket{0}\bra{0}_{E}).
\end{align}
In this case, the maximization over $\mathcal{R}^{\prime}$ and the minimization over $\rho$ decouple. Obviously, the maximization over $\mathcal{R}^{\prime}$ is achieved when $\mathcal{R}^{\prime}=I_E$. Because the final result does not depend on the choice of the pure dilation, we may take $\ket{\psi_\rho}$ as the canonical one. If $\rho$ can be diagonalized as $\rho=\sum_i p_i\ket{i_\rho}\bra{i_\rho}$, we choose $\ket{\psi_\rho}=\sum_i\sqrt{\rho}\ket{i_\rho}_S\ket{i_\rho}_P$. So
\begin{align}
    &\tr_{S} (V\ket{\psi_{\rho}}\bra{\psi_{\rho}}_{SP}V^{\dagger})=\sum_{i,j}\sqrt{p_i p_j}\bra{j_\rho} V^\dagger V\ket{i_\rho}\ket{i_\rho}\bra{j_\rho},\\
    &\tr_{S} (\ket{\psi_{\rho}}\bra{\psi_{\rho}}_{SP})=\rho_P=\sum_{k} p_k \ket{k_\rho}\bra{k_\rho},\\
    &\sqrt{\tr_{S} (\ket{\psi_{\rho}}\bra{\psi_{\rho}}_{SP})}\tr_{S} (V\ket{\psi_{\rho}}\bra{\psi_{\rho}}_{SP}V^{\dagger}) \sqrt{\tr_{S} (\ket{\psi_{\rho}}\bra{\psi_{\rho}}_{SP})} \notag \\
    =&\sum_{i,j,k,l}\sqrt{p_k} \ket{k_\rho}\bra{k_\rho}_P\sqrt{p_i p_j}\bra{j_\rho} V^\dagger V\ket{i_\rho} \ket{i_\rho}\bra{j_\rho}_P \sqrt{p_l} \ket{l_\rho}\bra{l_\rho}_P \\
    =&\sum_{i,j}p_i p_j \bra{j_\rho} V^\dagger V\ket{i_\rho}\ket{i_\rho}\bra{j_\rho}=(\rho V^\dagger V\rho)^\T.
\end{align}
After the calculation, we have the simplification 
\begin{eqs}\label{eq:Fminsimplifed}
    F_{\min}=\min_{\rho} \frac{\tr(\sqrt{(\rho V^\dagger V\rho)^\T})}{\sqrt{\tr(V\rho V^\dagger)}}=\min_{\rho} \frac{\tr(\sqrt{\rho V^\dagger V\rho})}{\sqrt{\tr(V\rho V^\dagger)}}.
\end{eqs}

Suppose $V$ has some zero singular values, or equivalently $\lambda_i=0$ for some $i$ in Eq.~(\ref{eq:V}), then the encoding $\tilde{\mathcal{E}}$ has a nontrivial kernel. In this case, we define $F_{\min}=0$, and it is achieved when $\rho$ is supported in $\ker V$. Indeed, such information is lost after encoding and hence cannot be recovered. We assume $\forall i, \lambda_i>0$ hereafter.

Using the Araki--Lieb--Thirring inequality \cite{araki1990inequality,audenaert2007araki}
\begin{equation}
\tr\big((A^r B^r A^r)^q \big) \leq\tr \big((ABA)^{rq} \big),
\end{equation}
for $r=\frac{1}{2}$, $q=1$, $A=\rho$ and $B=V^\dagger V$, we obtain $\tr(\rho\sqrt{V^\dagger V})\leq\tr(\sqrt{\rho V^\dagger V\rho})$. Equality is achieved if $[\rho, \Pi_\rho V^\dagger V\Pi_\rho]=0$, in which $\Pi_\rho$ is the support projector of $\rho$. Analyzing the quantity $\frac{\tr(\rho \sqrt{V^\dagger V})}{\sqrt{\tr(V\rho V^\dagger)}}$, it can be written as $\frac{\sum_{i=0}^{d-1} p_i \lambda_i^{\frac{1}{2}}}{\sqrt{\sum_{i=0}^{d-1} p_i \lambda_i}}$ in the diagonal basis of $V^\dagger V$, where $p_i\equiv\bra{v_i}\rho\ket{v_i}$ are diagonal entries of $\rho$ in the eigenbasis of $V^\dagger V$. Therefore, we have 
\begin{eqs}
    F_{\min}=\min_{\rho} \frac{\tr(\sqrt{\rho V^\dagger V\rho})}{\sqrt{\tr(V\rho V^\dagger)}}&\geq \min_{\rho} \frac{\tr(\rho \sqrt{V^\dagger V})}{\sqrt{\tr(V\rho V^\dagger)}}= \min_{\{p_i\}}\frac{\sum_{i=0}^{d-1} p_i \lambda_i^{\frac{1}{2}}}{\sqrt{\sum_{i=0}^{d-1} p_i \lambda_i}}.
\end{eqs}

Defining 
\begin{equation}
    \tilde f\equiv \frac{\sum_{i=0}^{d-1} p_i \lambda_i^{\frac{1}{2}}}{\sqrt{\sum_{i=0}^{d-1} p_i \lambda_i}},
\end{equation}
 we see that if $\{\lambda_i\}$ has degeneracy, then the change of the distribution of $\{p_i\}$ within one degenerate set of eigenvalues does not change the sum. Therefore, we can assign a single $p_i$ for one degenerate set, so we can write the sum from $0$ to $k$, with $k\leq d-1$, and we assume $\{\lambda_i\}$ are mutually distinct hereafter.  Because of the constraint $\sum_i p_i=1$, the $\{p_i\}$ are not independent. Without loss of generality, take $p_k=1-\sum_{j=0}^{k-1}p_j$. To find the saddle point, we do the variation with respect to the independent $p_i$'s. If this set of equations has a set of solution $\{p_i^{*}\}$, then the following explicit calculation implies that $\forall i=0,\ldots,k-1$, $\lambda_i=\lambda_j$. This contradicts our assumption that the $\{\lambda_i\}$ are mutually distinct, unless $k=1$.
\begin{align}
    &\frac{\partial\tilde f}{\partial p_i}=\frac{\left(\lambda_i^{\frac{1}{2}}-\lambda_k^{\frac{1}{2}}\right)\sqrt{\sum_{j}^k p_j \lambda_j}-\left(\lambda_i-\lambda_k\right)\frac{\sum_{j}^k p_j \lambda_j^{\frac{1}{2}}}{2\sqrt{\sum_{j}^k p_j \lambda_j}}}{\sum_{i=0}^k p_i \lambda_i}=0\\
    \Rightarrow&\quad  \lambda_i^{\frac{1}{2}}+\lambda_k^{\frac{1}{2}}=\frac{2 \sum_{j=0}^k p_j^* \lambda_j}{\sum_{j=0}^k p_j^* \lambda_j^{\frac{1}{2}}},~\forall i=0,\ldots,k-1.\\
    \Rightarrow&\quad  \lambda_i^{\frac{1}{2}}+\lambda_k^{\frac{1}{2}}=\lambda_j^{\frac{1}{2}}+\lambda_k^{\frac{1}{2}},~ \forall i,j=0,\ldots,k-1.
\end{align}

The above analysis implies that for $k>1$, the saddle points cannot be in the interior of the region $0<p_i<1$. The extremal points have to be on the boundary, which means some of the $p_i$'s are $0$. This effectively reduces $k$ to a lower dimension $k^\prime$. If $k^\prime>1$, we can reiterate the above analysis and further reduce $k^\prime$ to $1$. For the two-dimensional case, we can calculate the minimum explicitly. 
\begin{align}
    &\tilde f=\frac{p\lambda_0^{\frac{1}{2}}+(1-p)\lambda_1^{\frac{1}{2}}}{\sqrt{p\lambda_0+(1-p)\lambda_1}},\\
    &\frac{\partial\tilde f}{\partial p_i}=0\quad\Rightarrow\quad p^*=\frac{ \lambda_1^{\frac{1}{2}}}{ \lambda_0^{\frac{1}{2}}+ \lambda_1^{\frac{1}{2}}},\\
    &\tilde f (p=p^*)=\frac{2\lambda_0^{\frac{1}{4}}\lambda_1^{\frac{1}{4}}}{\lambda_0^{\frac{1}{2}}+\lambda_1^{\frac{1}{2}}}.
\end{align}
This quantity depends only on the ratio $\frac{\lambda_0}{\lambda_1}$ and is symmetric under the exchange of $\lambda_0$ and $\lambda_1$. Therefore we can assume $\frac{\lambda_0}{ \lambda_1}<1$ and the minimum of $\tilde f$ is achieved with minimal $\frac{\lambda_0}{\lambda_1}$. Because $\lambda_{\min} \leq\lambda_i\leq\lambda_{\max}$, we finally have
\begin{eqs}
    F_{\min}=\min_{\{p_i\}}\tilde f=\frac{2\lambda_{\max}^{\frac 14}\lambda_{\min}^{\frac 14}}{\lambda_{\max}^{\frac 12}+\lambda_{\min}^{\frac 12}}=\frac{2r^{\frac{1}{4}}}{1+r^{\frac{1}{2}}},
\end{eqs}
as claimed.

When $\lambda_{\min}=0$, the above formula gives $F_{\min}=0$, which is consistent with the previous analysis. The final remaining case is $k=0$. In this case, we always have $\tilde f=1$. This happens when all eigenvalues $\{\lambda_i\}$ are identical or when $\rho$ is a pure state. The former case is consistent with the formula above. The latter case gives the maximum of $\tilde f$, which is not relevant for our purposes.

\subsection{Choi fidelity}
Using Theorem~\ref{theorem:fidelitycomplementary} 
for the case $\mathcal{\overline N}=\tilde{\mathcal{E}}$, $\mathcal{\overline M}=\mathrm{id}$ and $\mathcal{\overline M}_c=\tr_S$,
the calculation of $F_{\mathrm{Choi}}$ is similar to that of $F_{\min}$. We repeat the procedure and achieve a similar equation to Eq.~(\ref{eq:Fminsimplifed}).

\begin{eqs}
    F_{\mathrm{Choi}}&= \frac{\tr(\sqrt{\rho_{\mathrm{Choi}} V^\dagger V\rho_{\mathrm{Choi}}})}{\sqrt{\tr(V\rho_{\mathrm{Choi}} V^\dagger)}}=\frac{\tr(\sqrt{V^\dagger V})}{\sqrt{d \cdot \tr(V^\dagger V)} }=\frac{\sum_{i=0}^{d-1}\lambda^{\frac{1}{2}}_i}{\sqrt{d}(\sum_{i=0}^{d-1}\lambda_i)^{\frac{1}{2}}}.
\end{eqs}
Here we use $\rho_{\mathrm{Choi}}=\tr_P (\ket{\psi_c}\bra{\psi_c})=\frac{\mathbbm{1}_L}{d}$, where $\ket{\psi_c}=\sum_{i=0}^{d-1} \frac{1}{\sqrt{d}}\ket{\overline i}\ket{\overline i}$.

From Eq.~\eqref{eq:FChoi},
\begin{eqs}
    F_{\mathrm{Choi}}^2&=\max_{\{E_k\}}\sum_k\frac{\bra{\psi_c}E_k V\ket{\psi_c}\bra{\psi_c}V^\dagger E_k^\dagger\ket{\psi_c}}{\bra{\psi_c}V^\dagger V\ket{\psi_c}}=\max_{\{E_k\}}\frac{1}{d \cdot \tr(V^\dagger V)}\sum_k \tr(E_k V) \tr(V^\dagger E_k^\dagger),
\end{eqs}
we see that the optimal recovery channel is achieved when $E_1=\sum_i \ket{v_i}\bra{\tilde i_p} U^\dagger$ and other Kraus operators $E_i V=0$ to satisfy $\sum_k E_k^\dagger E_k=\mathbbm{1}_P$, where $\mathbbm{1}_P$ is the identity matrix in physical Hilbert space $\mathcal{H}_P$. Here $V$ is represented as a polar decomposition $V=\sum_i \lambda_i^{1/2} U \ket{\tilde i_p}\bra{v_i}$, with $U$  unitary and $\{\ket{\tilde i_p}\}$ a set of orthonormal bases in the physical Hilbert space.

\subsection{Average fidelity}
Writing out Eq.~\eqref{eq:Favg} explicitly, it is
\begin{eqs}
    F_{\mathrm{avg}}^2&=\max_{\{E_k\}}\int_{\mathrm{Haar}} d\psi \sum_k\frac{\bra{\psi}E_k V\ket{\psi}\bra{\psi}V^\dagger E_k^\dagger\ket{\psi}}{\bra{\psi}V^\dagger V\ket{\psi}}=\max_{\{E_k\}} \tr (O_E O_{\psi}),
\end{eqs}
where $\int d\psi$ is an integral over the Haar measure in the logical Hilbert space, and we have defined  
\begin{eqs}
    O_E\equiv\sum_k E_k V\otimes V^\dagger E_k^\dagger,~O_{\psi}\equiv\int_{\mathrm{Haar}} d\psi \frac{\ket{\psi}\bra{\psi}\otimes\ket{\psi}\bra{\psi}}{\bra{\psi}V^\dagger V\ket{\psi}}.
\end{eqs}

We first calculate $O_{\psi}$. Under the basis of the eigenstates of $V^\dagger V$, the Haar measure can be written as an integral over the amplitudes and phases of each basis. Specifically,
\begin{eqs}\label{eq:Haarbasis}
    \int_{\mathrm{Haar}} d\psi=\frac{(d-1)!}{(2\pi)^d}\int d p_0 \cdots d p_{d-1} \int d \theta_0 \cdots d \theta_{d-1} ~\delta\left(\sum_{i=0}^{d-1} p_i -1\right),
\end{eqs}
where $(d-1)!$ is the normalization constant from the integral of the standard $(d-1)$-simplex. So after integrating out the phases, we obtain
\begin{eqs}
    &O_{\psi}=(d-1)!\int d p_0 \cdots d p_{d-1}  ~\delta\left(\sum_{l=0}^{d-1} p_l -1\right)\frac{\sum_{i=0}^{d-1} p_i^2 \ket{v_i v_i}\bra{v_i v_i}+\sum_{i\neq j}^{d-1} p_i p_j (\ket{v_i v_j}\bra{v_i v_j}+\ket{v_i v_j}\bra{v_j v_i})}{\sum_{k=0}^{d-1} \lambda_k p_k}.
\end{eqs}
Integrating out the $p_i$ dependent part, we obtain the explicit form of the matrix coefficient of $O_{\psi}$,
\begin{eqs}
    O_{ij}=&(d-1)!\int d p_0 \cdots d p_{d-1}  ~\delta\left(\sum_{l=0}^{d-1} p_l -1\right)\frac{ p_i p_j }{\sum_{k=0}^{d-1} \lambda_k p_k}\\
     =& \frac{\partial}{\partial \lambda_i}\frac{\partial}{\partial \lambda_j}\left( \frac{1}{d} \sum_{n=0}^{d-1} \frac{\lambda_n^d \log\lambda_n}{\prod_{m\neq n}(\lambda_n -\lambda_m)} \right),\\
     \equiv & \frac{1}{d}\lambda^d \log\lambda [\lambda_0, \ldots , \lambda_{d-1},\lambda_i,\lambda_j].\\
\end{eqs}
This is known as the divided difference of the function $\frac{1}{d}\lambda^d \log\lambda$. It is valid for any $i$, $j$, and the case of degenerate eigenvalues is defined via taking the limit.

Now we turn back to the maximization of $F_{\mathrm{avg}}^2$. With the explicit form of $O_{ij}$, we now have
\begin{eqs}\label{eq:Favgcoefficent}
    F_{\mathrm{avg}}^2=&\max_{\{E_k\}}\sum_k \Bigg[\sum_{i=0}^{d-1} O_{ii} \bra{v_i}E_k V\ket{v_i}\bra{v_i}V^\dagger E_k^\dagger\ket{v_i} +
    \sum_{i\neq j}^{d-1} O_{ij} \Big(\bra{v_i}E_k V\ket{v_i}\bra{v_j}V^\dagger E_k^\dagger\ket{v_j}
    + \bra{v_j}E_k V\ket{v_i}\bra{v_i}V^\dagger E_k^\dagger\ket{v_j}\Big ) \Bigg]\\
    =&\max_{\{E_k\}}\sum_k \Bigg(\sum_{i=0}^{d-1} \lambda_i O_{ii} E_{k,ii} E_{k,ii}^\dagger + \sum_{i\neq j}^{d-1} \sqrt{\lambda_i\lambda_j} O_{ij} E_{k,ii} E_{k,jj}^\dagger+ \sum_{i\neq j}^{d-1} \lambda_i O_{ij} E_{k,ji} E_{k,ij}^\dagger \Bigg).
\end{eqs}
Similar to the explicit recovery channel of $F_{\mathrm{Choi}}$, we use the polar decomposition $V=\sum_i \lambda_i^{1/2} U \ket{\tilde i_p}\bra{v_i}$, with $U$  unitary and define $E_k=\sum_{i,j} E_{k,ij}\ket{v_i}\bra{\tilde j_p} U^\dagger$. Though we may need to add extra Kraus operators to satisfy $\sum_k E_k^\dagger E_k=\mathbbm{1}_P$ in the sum, we only need to consider the entries contributing to $E_k V\neq 0$ and treat each $E_k$ as a square matrix. 

We first show that $\forall i\neq j$, $O_{ij}\leq O_{ii}$. By renaming the integral variable, we have
\begin{align}
    O_{ii}-O_{ij}=&\frac{1}{2}\int d\mu(p) \Big(\frac{p_i^2}{\lambda_i p_i+\lambda_j p_j+R}+\frac{p_j^2}{\lambda_i p_j+\lambda_j p_i+R}-\frac{p_i p_j}{\lambda_i p_i+\lambda_j p_j+R}-\frac{p_i p_j}{\lambda_i p_j+\lambda_j p_i+R} \Big)\\
    =&\frac{1}{2}\int d\mu(p) \frac{(p_i-p_j)^2 \big(R+\lambda_j (p_i+p_j) \big)}{(\lambda_i p_i+\lambda_j p_j+R)(\lambda_i p_j+\lambda_j p_i+R)}\geq 0,
\end{align}
where $R=\sum_{k\neq i,j}\lambda_k p_k$ and $d\mu (p)$ is a shorthand notation for the $p$ integral in Eq.~\eqref{eq:Haarbasis}. The difference is greater than zero because the integrand is manifestly non-negative. 

The trace-preserving constraint for Kraus operators can be written as 
\begin{align}\label{eq:Krausconstraint}
    \sum_{k}\sum_{j} E^{\dagger}_{k,ij} E_{k,ji}=\sum_{k}\sum_{j} |E_{k,ji}|^2=1,~\forall i.
\end{align}
Therefore, combining the first and the third terms in Eq.~\eqref{eq:Favgcoefficent}, we have
\begin{align}
    \sum_k \left(\sum_{i=0}^{d-1} \lambda_i O_{ii} E_{k,ii} E_{k,ii}^\dagger + \sum_{i\neq j}^{d-1}   \lambda_i O_{ij} E_{k,ji} E_{k,ij}^\dagger \right)
    \leq &\sum_k \left(\sum_{i=0}^{d-1} \lambda_i O_{ii} E_{k,ii} E_{k,ii}^\dagger + \sum_{i\neq j}^{d-1}   \lambda_i O_{ii} E_{k,ji} E_{k,ij}^\dagger \right)\\
    =& \sum_k \left(\sum_{i=0}^{d-1} \lambda_i O_{ii} \sum_{j} E_{k,ji} E_{k,ij}^\dagger \right)= \sum_{i=0}^{d-1} \lambda_i O_{ii}.
\end{align}
Equality is achieved when each $E_k$ only has diagonal entries.

For the second term, we have
\begin{align}
    \sum_k\sum_{i\neq j}^{d-1} \sqrt{\lambda_i\lambda_j} O_{ij} E_{k,ii} E_{k,jj}^\dagger\leq& \sum_{i\neq j}^{d-1} \sqrt{\lambda_i\lambda_j} O_{ij}\left(\sum_k |E_{k,ii}|^2\right)^{1/2} \left(\sum_k |E_{k,jj}|^2\right)^{1/2}\\
    \leq &\sum_{i\neq j}^{d-1} \sqrt{\lambda_i\lambda_j} O_{ij}.
\end{align}
The first inequality is a Cauchy inequality for the sum over Kraus operators. The second inequality comes from the constraint Eq.~\eqref{eq:Krausconstraint}. Equality is achieved when each $E_k$ only has diagonal entries and $E_{k,ii}=E_{k,jj},~\forall i,j$. In this case, we can simplify to have a single Kraus operator with ${E_{1,ii}}=1$.

Combining the above analysis, we obtain 
\begin{eqs}
    F_{\mathrm{avg}}^2&\leq \sum_{i,j=0}^{d-1} \sqrt{\lambda_i\lambda_j} O_{ij}=\sum_{i,j=0}^{d-1} \sqrt{\lambda_i\lambda_j} \frac{\partial}{\partial \lambda_i}\frac{\partial}{\partial \lambda_j}\left( \frac{1}{d} \sum_{n=0}^{d-1} \frac{\lambda_n^d \log\lambda_n}{\prod_{m\neq n}(\lambda_n -\lambda_m)} \right).
\end{eqs}
The inequality is achieved when there is a single Kraus operator $E_1=\sum_i \ket{v_i}\bra{\tilde i_p} U^\dagger$. This is the same recovery channel as the case for $F_{\mathrm{Choi}}$.

\subsection{An inequality between average and Choi optimal fidelities under noise}\label{sec:avgvschoi}
Here we prove a simple relation that $F_{\mathrm{avg}}^{(\mathcal{N})2}\geq F_{\mathrm{Choi}}^{(\mathcal{N})2}$, which implies $\epsilon_{\mathrm{Choi}}^{(\mathcal{N})}\geq\epsilon_{\mathrm{avg}}^{(\mathcal{N})}$.

\begin{proof}
Suppose the error channel $\mathcal{N}$ can be written in terms of the Kraus operators $\mathcal{N}(\cdot)=\sum_l N_l(\cdot)N_l^\dagger$. For $F_{\mathrm{Choi}}^{(\mathcal{N})2}$, we assume that the optimal recovery channel  is achieved with a set of Kraus operators $\mathcal{R}_c(\cdot)=\sum_k E_{c,k}(\cdot)E_{c,k}^\dagger$ and write 
\begin{eqs}
   F_{\mathrm{Choi}}^{(\mathcal{N})2}
   &= \sum_k \sum_l\frac{\bra{\psi_c}E_{c,k} N_l V\ket{\psi_c}\bra{\psi_c}V^\dagger N_l^\dagger E_{c,k}^\dagger\ket{\psi_c}}{\bra{\psi_c}V^\dagger V\ket{\psi_c}}\\
   &=\sum_{k} \sum_l \frac{1}{d\tr(V^\dagger V)}\tr(E_{c,k} N_l V) \tr(V^\dagger N_l^\dagger E_{c,k}^\dagger)\\
   &=\sum_{k} \sum_l \frac{d}{\tr(V^\dagger V)}\left |\int_{\mathrm{Haar}} d\psi \bra{\psi} E_{c,k} N_l V \ket{\psi} \right|^2\\
   &\equiv \sum_{k} \sum_l \frac{\tr(V^\dagger V)}{d} \left|\int d\nu(\psi)  O_{c,kl}(\psi)\right|^2.
\end{eqs}
Here, because $\int_{\mathrm{Haar}} d\psi \bra{\psi}V^\dagger V\ket{\psi}=\frac{\tr(V^\dagger V)}{d}$, we can define a new measure $d\nu(\psi)\equiv \frac{d}{\tr(V^\dagger V)}d\psi \bra{\psi}V^\dagger V\ket{\psi}$. We also introduced the shorthand notation $O_{c,kl}(\psi)\equiv \frac{\bra{\psi} E_{c,k} N_l V \ket{\psi}}{\bra{\psi}V^\dagger V\ket{\psi}}$. For each $k$ and $l$, we have 
\begin{eqs}
    \int d\nu(\psi)  |O_{c,kl}(\psi)|^2 \geq \left|\int d\nu(\psi)  O_{c,kl}(\psi)\right|^2.
\end{eqs}
Therefore, we have 
\begin{eqs}
   F_{\mathrm{Choi}}^{(\mathcal{N})2}&\leq  \sum_{k} \sum_l \frac{\tr(V^\dagger V)}{d} \int d\nu(\psi)  |O_{c,kl}(\psi)|^2\\
   &= \sum_{k} \sum_l \int_{\mathrm{Haar}} d\psi  \frac{|\bra{\psi} E_{c,k} N_l V \ket{\psi}|^2}{\bra{\psi}V^\dagger V\ket{\psi}} \\
    &\leq \max_{\{E_k\}} \sum_{k} \sum_l \int_{\mathrm{Haar}} d\psi  \frac{|\bra{\psi} E_{k} N_l V \ket{\psi}|^2}{\bra{\psi}V^\dagger V\ket{\psi}}\\
    &=F_{\mathrm{avg}}^{(\mathcal{N})2}.
\end{eqs}
The last inequality follows from the fact that $\mathcal{R}_c$ may not be the optimal recovery channel for the maximization of $F_{\mathrm{avg}}^{(\mathcal{N})2}$.

\end{proof}

\subsection{Proof of Proposition \ref{prop:optimalFnoise}}\label{sec:optimalFineqproof}

\begin{proof}
(a):
Suppose that the optimal recovery channels of $F_{\mathrm{min}}$ and $F_{\mathrm{min}}^{(\mathcal{N})}$ are attained by $\mathcal{R}_o$ and $\mathcal{R}_o^{(\mathcal{N})}$, respectively. Then $\mathcal{R}_o^{(\mathcal{N})}\circ \mathcal{N}$ can be seen as a particular recovery channel of $\tilde{\mathcal{E}}$. Since this candidate channel is not necessarily optimal compared with $\mathcal{R}o$,  we have $F_{\mathrm{min}}^{(\mathcal{N})}\leq F_{\mathrm{min}}$ and hence $\epsilon_{\mathrm{min}}^{(\mathcal{N})}\geq \epsilon_{\mathrm{min}}$. Similarly, we can prove that $\epsilon_{\mathrm{Choi}}^{(\mathcal{N})}\geq \epsilon_{\mathrm{Choi}}$ and $\epsilon_{\mathrm{avg}}^{(\mathcal{N})}\geq \epsilon_{\mathrm{avg}}$.

For any state $\rho\in\mathcal{H}_L$, and arbitrary channels $ \mathcal{R}_{\mathcal{N}}:\mathcal{H}_P\to\mathcal{H}_P$, $\mathcal{R}_{i}:\mathcal{H}_P\to\mathcal{H}_L$, 
\begin{eqs}
    \min_{\mathcal{R}} \epsilon_{\rho}(\mathcal{R}\circ\mathcal{N}\circ \tilde{\mathcal{E}},\mathrm{id})&\leq \epsilon_{\rho}(\mathcal{R}_i\circ\mathcal{R_{\mathcal{N}}}\circ\mathcal{N}\circ \tilde{\mathcal{E}},\mathrm{id})\leq \epsilon_{\rho}(\mathcal{R}_i\circ\mathcal{R_{\mathcal{N}}}\circ\mathcal{N}\circ \tilde{\mathcal{E}},\mathcal{R}_i\circ \tilde{\mathcal{E}}) +\epsilon_{\rho}(\mathcal{R}_i\circ \tilde{\mathcal{E}},\mathrm{id})\\&\leq \epsilon_{\rho}(\mathcal{R_{\mathcal{N}}}\circ\mathcal{N}\circ \tilde{\mathcal{E}},\tilde{\mathcal{E}}) +\epsilon_{\rho}(\mathcal{R}_i\circ \tilde{\mathcal{E}},\mathrm{id}).
\end{eqs}
Because $\mathcal{R}_{\mathcal{N}}$ and $\mathcal{R}_i$ are totally arbitrary, we can minimize over them separately to get a tighter bound.

Taking the maximum over $\rho$ on both sides, followed by the minimization over $\mathcal{R}_{\mathcal{N}}$ and $\mathcal{R}_i$, yields (a). The legitimacy of interchange $\max_{\rho}$ and $\min_{\mathcal{R}}$ is proven in Appendix~\ref{sec:theoremproof}. If we take $\rho=\ket{\psi_c}\bra{\psi_c}$ as the Choi state and minimize over $\mathcal{R}_{\mathcal{N}}$ and $\mathcal{R}_i$, we obtain (a) for the Choi quantities. For the average case, we note that $\epsilon^{(\mathcal{N})}_{\mathrm{avg}}\leq \epsilon_{\mathrm{avg}}(\mathcal{R}_i\circ\mathcal{R_{\mathcal{N}}}\circ\mathcal{N}\circ \tilde{\mathcal{E}},\mathrm{id})$. If we first integrate the inequality $\epsilon_{\rho}(\mathcal{R}_i\circ\mathcal{R_{\mathcal{N}}}\circ\mathcal{N}\circ \tilde{\mathcal{E}}, \mathrm{id})\leq \epsilon_{\rho}(\mathcal{R_{\mathcal{N}}}\circ\mathcal{N}\circ \tilde{\mathcal{E}},\tilde{\mathcal{E}}) +\epsilon_{\rho}(\mathcal{R}_i\circ \tilde{\mathcal{E}}, \mathrm{id})$ over the Haar measure of pure states and then minimize over $\mathcal{R}_{\mathcal{N}}$ and $\mathcal{R}_i$, we obtain (a) for the average quantities.

(b): Follows immediately from (a).

(c): It is easy to see $\min_{\mathcal{R}_{\mathcal{N}}}\epsilon_{\mathrm{min}}(\mathcal{R_{\mathcal{N}}}\circ\mathcal{N}\circ \tilde{\mathcal{E}},\tilde{\mathcal{E}})=0$ is equivalent to $\min_{\mathcal{R}_{\mathcal{N}}}\epsilon_{\mathrm{Choi}}(\mathcal{R_{\mathcal{N}}}\circ\mathcal{N}\circ \tilde{\mathcal{E}},\tilde{\mathcal{E}})=0$ and $\min_{\mathcal{R}_{\mathcal{N}}}\epsilon_{\mathrm{avg}}(\mathcal{R_{\mathcal{N}}}\circ\mathcal{N}\circ \tilde{\mathcal{E}},\tilde{\mathcal{E}})=0$. Using Theorem~\ref{theorem:fidelitycomplementary} for the case $(n_1,n_2)=(1,1)$, $\mathcal{N}_1=\mathcal{M}_1=\mathrm{id}_L$, $\mathcal{N}_2=\mathcal{N}$ and $\mathcal{M}_2=\mathrm{id}_P$, we have

\begin{align}
&(\mathcal{R}^{\prime}\circ\tilde{\mathcal{E}}_c)\otimes \mathrm{id}(\ket{\psi_{\rho}}\bra{\psi_{\rho}}_{SP})=\frac{\tr_{S} (V\ket{\psi_{\rho}}\bra{\psi_{\rho}}_{SP}V^{\dagger})}{\tr(V\rho V^\dagger)}\otimes\mathcal{R}^{\prime}(\ket{0}\bra{0}_{E}),\\
& (\mathcal{N}\circ\tilde{\mathcal{E}})_c\otimes \mathrm{id}(\ket{\psi_{\rho}}\bra{\psi_{\rho}}_{SP})=\sum_{i,j}\frac{\tr_{S} (N_i V\ket{\psi_{\rho}}\bra{\psi_{\rho}}_{SP}V^{\dagger}N_j^\dagger)}{\tr(V\rho V^\dagger )}\otimes \ket{i}\bra{j}_{E}.
\end{align}
The channel $\mathcal{R}^{\prime}(\ket{0}\bra{0}_{E})$ always maps to a fixed state $\sigma_{0E}$. The condition $\min_{\mathcal{R}_{\mathcal{N}}}\epsilon_{\mathrm{min}}(\mathcal{R_{\mathcal{N}}}\circ\mathcal{N}\circ \tilde{\mathcal{E}},\tilde{\mathcal{E}})=0$ holds if and only if the above two states are identical for any state $\ket{\psi_\rho}$. This is only possible when $V^\dagger N^\dagger_i N_j V=\Lambda_{ij} V^\dagger V,~ \forall~i,j,$ and we can choose $\sigma_{0E}$ such that $\bra{i}\sigma_{0E}\ket{j}=\Lambda_{ij}$.

\end{proof}

\section{Calculations related to logical operators}\label{sec:callogical}

We first show that for an isometric encoding $V_0$ mapping from a finite dimensional logical Hilbert space $\mathcal{H}_L$ to a physical Hilbert space $\mathcal{H}_P$, unitary operators $U_L$, $U_P$ on $\mathcal{H}_L$ and $\mathcal{H}_P$ respectively,
\begin{eqs}
  U_P V_0= V_0 U_L \quad \Leftrightarrow \quad V_0^\dagger  U_P V_0= U_L .
\end{eqs}
\begin{proof}
$ U_P V_0= V_0 U_L\Rightarrow V_0^\dagger  U_P V_0= U_L$ is obvious, by multiplying $V_0^\dagger $ from the left on both sides.

To show $V_0^\dagger  U_P V_0= U_L \Rightarrow U_P V_0= V_0 U_L$, we multiply $V_0$ from the left. Denoting $V_0V_0^\dagger\equiv P$, we have $P U_P V_0=V_0 U_L$. Multiply this equation by its complex conjugate from the right, we obtain $PU_PPU_P^\dagger P=P$. Note that $U_PPU_P^\dagger\equiv Q$ is also a projection operator. For any state $\ket{\psi}\in P\mathcal{H}_P$, 
\begin{eqs}
    \bra{\psi}P Q P\ket{\psi}=\bra{\psi}Q\ket{\psi}=\bra{\psi}P\ket{\psi}=\bra{\psi}\psi\rangle.
\end{eqs}
Therefore, if $\ket{\psi}\in P\mathcal{H}_P$, then $\ket{\psi}\in Q\mathcal{H}_P$. So $Q\geq P$. But $U_PPU_P^\dagger\equiv Q$. This means for a finite-dimensional logical Hilbert space, $P=Q=U_PPU_P^\dagger$. So $U_P P=P U_P$. Getting back to $P U_P V_0=V_0 U_L$, we have $P U_P V_0=U_P PV_0=U_P V_0=V_0 U_L$.

\end{proof}

\subsection{Proposition \ref{prop:logical2}}\label{sec:callogical2}

We then present the details for Proposition \ref{prop:logical2}. The calculation is overall similar to that of $F_{\mathrm{min}}$. We use Theorem~\ref{theorem:fidelitycomplementary} for the case $(n_1,n_2)=(1,1)$, $\mathcal{N}_1=\mathrm{id}_L$, $\mathcal{N}_2=\mathrm{id}_P$, $\mathcal{M}_1=\mathcal{U}_L$, $\mathcal{M}_2=\mathrm{id}_P$.

\begin{align}
&\tilde{\mathcal{E}}_c\otimes \mathrm{id}(\ket{\psi_{\rho}}\bra{\psi_{\rho}}_{SP})=\frac{\tr_{S} (V\ket{\psi_{\rho}}\bra{\psi_{\rho}}_{SP}V^{\dagger})}{\tr(V\rho V^\dagger)}\otimes\ket{0}\bra{0}_E,\\
&(\mathcal{R}^{\prime})\circ(\tilde{\mathcal{E}}\circ \mathcal{U}_L)_c\otimes \mathrm{id} (\ket{\psi_{\rho}}\bra{\psi_{\rho}}_{SP})=\frac{\tr_{S} (VU_L\ket{\psi_{\rho}}\bra{\psi_{\rho}}_{SP}U_L^\dagger V^{\dagger})}{\tr(VU_L\rho U_L^\dagger V^\dagger)}\otimes \mathcal{R}^{\prime}(\ket{0}\bra{0}_{E}).
\end{align}

Obviously, the maximization over $\mathcal{R}^{\prime}$ is achieved when $\mathcal{R}^{\prime}=I_E$. Here we choose $\ket{\psi_\rho}=\sum_i\sqrt{\rho}\ket{i_\rho}_S\ket{i_\rho}_P$ as before. So
\begin{align}
    &\tr_{S} (V\ket{\psi_{\rho}}\bra{\psi_{\rho}}_{SP}V^{\dagger})=\sum_{i,j}\bra{j_\rho}\rho^{1/2} V^\dagger V  \rho^{1/2}\ket{i_\rho}\ket{i_\rho}\bra{j_\rho}=(\rho^{1/2} V^\dagger V  \rho^{1/2})^\T,\\
    &\tr_{S} (VU_L\ket{\psi_{\rho}}\bra{\psi_{\rho}}_{SP}U_L^\dagger V^{\dagger})=\sum_{i,j}\bra{j_\rho}\rho^{1/2} U_L^\dagger V^\dagger V  U_L \rho^{1/2}\ket{i_\rho}\ket{i_\rho}\bra{j_\rho}
    =(\rho^{1/2} U_L^\dagger V^\dagger V  U_L \rho^{1/2})^\T.
\end{align}
After calculation, we obtain Eq.~\eqref{eq:fidelityUL}, by noting $\sqrt{A^\T}=(\sqrt{A})^\T$ for any positive semi-definite operator.

If we require the logical channel with respect to $\mathcal{U}_L$ to be a unitary channel $\mathcal{U}_P$ on the physical Hilbert space, we can calculate the maximal fidelity as follows.

\begin{align}\label{eq:fidelityULunitary}
\max_{\mathcal{U}_P}F_{\rho}&=\max_{U_P}\frac{|\bra{\psi_\rho}V^\dagger U_P^\dagger V U_L\ket{\psi_\rho}|}{\sqrt{\tr(\rho V^\dagger V)\tr(\rho U_L^\dagger V^\dagger V U_L)}}\\&=\max_{U_P}\frac{|\tr( \rho V^\dagger U_P^\dagger  V U_L)|}{\sqrt{\tr(\rho V^\dagger V)\tr(\rho U_L^\dagger V^\dagger V U_L)}}\\
 &=\frac{\tr\left(\sqrt{ V \rho U_L^\dagger  V^\dagger V U_L\rho V^\dagger}\right)}{\sqrt{\tr(\rho V^\dagger V)\tr(\rho U_L^\dagger V^\dagger V U_L)}}.
\end{align}
The last equality follows from the standard maximization identity $\max_U |\tr(OU)|=\tr(|O|)$, with $O= V U_L\rho V^\dagger$ in this case. Since $V U_L\rho V^\dagger$ does not have full support on the physical Hilbert space, $U_P$ is only determined on its support as the unitary of the polar decomposition of $V U_L\rho V^\dagger$. On the subspace of the kernel of $V U_L\rho V^\dagger$, $U_P$ can be chosen arbitrarily. Comparing Eq.~\eqref{eq:fidelityULunitary} with Eq.~\eqref{eq:fidelityUL}, if we denote $\rho^{1/2} U_L^\dagger V^\dagger V  U_L \rho^{1/2}=B$
 and  $ V \rho^{1/2}=A$ for the moment, we see that the result of Eq.~\eqref{eq:fidelityULunitary} is proportional to $\tr(\sqrt{ABA^\dagger})$ while Eq. \eqref{eq:fidelityUL} is proportional to $\tr(\sqrt{|A|B|A|})$. Let the singular value decomposition of $A$ be $A=S\Lambda T$, where $\Lambda$ is a diagonal matrix with positive singular values, but may be rank or column deficient.  Then $|A|=\sqrt{A^\dagger A}=T^\dagger |\Lambda| T$. Therefore,
$\tr(\sqrt{|A|B|A|})=\tr(\sqrt{|\Lambda|TBT^\dagger|\Lambda|})$, while $\tr(\sqrt{ABA^\dagger})=\tr(\sqrt{\Lambda TBT^\dagger \Lambda})$. Here $\Lambda$ and $|\Lambda|$ are both diagonal matrices with the same diagonal entries but with different shapes. A direct comparison shows that the two expressions have the same nonzero spectrum, regardless of whether $A$ is rank or column deficient. Hence the two trace quantities are equal.

Therefore, we have shown that there exists a unitary channel $\mathcal{U}_P$ that achieves the optimal channel fidelity to represent the channel $\mathcal{U}_L$.

\subsection{Proposition \ref{prop:bestworstUL}} \label{sec:calbestworstUL}

Eq.~\eqref{eq:FChoilogical} can be viewed as the fidelity between two density matrices $\rho_1=\frac{V^\dagger V}{d}$ and $\rho_2=U_L^\dagger \rho_1 U_L$. Therefore, the maximal fidelity is $1$ and is achieved if and only if $\rho_1=\rho_2$, thus $[U_L, V^\dagger V]=0$. For the minimal fidelity, we have 
\begin{equation}
    F(\rho_1,\rho_2)=\|\sqrt{\rho_1}\sqrt{\rho_2} \|_1\geq\tr(\sqrt{\rho_1}\sqrt{\rho_2})=\frac{1}{d}\tr(\sqrt{V^\dagger V}U_L^\dagger\sqrt{V^\dagger V}U_L).
\end{equation}
In the eigen-basis of $V^\dagger V$, we have 
\begin{equation}
    \frac{1}{d}\tr(\sqrt{V^\dagger V}U_L^\dagger\sqrt{V^\dagger V}U_L)=\frac{1}{d}\sum_{i,j}\sqrt{\lambda_i \lambda_j} |(U_L)_{ij}|^2\equiv \frac{1}{d}\vec{v}\cdot W\cdot \vec{v}^\T,
\end{equation}
in which $\vec{v}\equiv(\sqrt{\lambda_0},\ldots,\sqrt{\lambda_{d-1}})$ and $W$, defined via  $W_{ij}\equiv|U_L|_{ij}^2$, is a doubly stochastic matrix. Because a doubly stochastic matrix can be written as a linear combination of permutation matrices, we have $\min_W\vec{v}\cdot W\cdot \vec{v}^\T=\sum_i \sqrt{\lambda_{(i)}\lambda_{(d-1-i)}}$. The minimum is achieved if we take $U_{L,\mathrm{min}}$ to act on each eigenstate of $V^\dagger V$ as $U_{L,\mathrm{min}}\ket{v_{(j)}}=e^{i\theta_j}\ket{v_{(d-1-j)}}$, as stated in Proposition \ref{prop:bestworstUL}.

\section{Relations to other measures of non-isometry}\label{sec:comparison}

In this part, we set $\tr (V^\dagger V)=d$, which greatly simplifies the calculation. Ref.~\cite{antonini2025non} defines the following quantity to quantify the non-isometry 
\begin{align}
    D(V)&=\int_{\mathrm{Haar}} d\psi d\phi |\bra{\psi}(V^\dagger V-\mathbbm{1}_L) \ket{\phi}|^2=\frac{1}{d^2}\tr\left((V^\dagger V-\mathbbm{1}_L)^2\right)=\frac{1}{d^2}(\|V\|_4^4-d).
\end{align}
We can use the interpolation inequality of the Schatten $p$-norm
\begin{align}
    \|V\|_r^\theta \|V\|_s^{1-\theta}\geq \|V\|_p, ~\frac{1}{p}=\frac{\theta}{r}+\frac{1-\theta}{s},
\end{align}
for the case $\theta =\frac{1}{3}$, $s=4$, $p=2$,$r=1$ to get $\|V\|_1 \|V\|_4^2\geq \|V\|_2^3$. Therefore,
\begin{align}
    F_{\mathrm{Choi}}=\frac{1}{\sqrt{d}}\frac{\|V\|_1}{\|V\|_2}\geq \frac{1}{\sqrt{d}}\frac{\|V\|_2^2}{\|V\|_4^2}=\frac{1}{\sqrt{1+d D(V)}}.
\end{align}
Written in terms of $\epsilon_{\mathrm{Choi}}$, we have 
\begin{align}
    \epsilon_{\mathrm{Choi}}^2=1-F_{\mathrm{Choi}}^2\leq \frac{d D(V)}{1+d D(V)}\leq d D(V).
\end{align}

Ref.~\cite{antonini2025non} uses the following quantity to quantify the deviation from the exact state-independent reconstruction of $U_L$

\begin{align}
    &\min_{U_P} \int_{\mathrm{Haar}} d\psi \|(U_P V- V U_L) \ket{\psi}\|^2=2-\max_{U_P}\mathrm{Re}\left(\frac{2}{d}\tr(V^\dagger U_P V U_L)\right).
\end{align}
Noting that 
\begin{align}
    &\max_{U_P} \mathrm{Re}\left(\tr(V^\dagger U_P V U_L)\right)= \max_{U_P} |\tr(V^\dagger U_P V U_L)|= \tr(\sqrt{ V U_L^\dagger V^\dagger V U_L V^\dagger }),
\end{align}
we immediately see that 
\begin{align}
    &\min_{U_P} \int_{\mathrm{Haar}} d\psi \|(U_P V- V U_L) \ket{\psi}\|^2= 2-2 \max_{\mathcal{R}}F_{\mathrm{Choi}}(\mathcal{R}\circ\tilde{\mathcal{E}},\tilde{\mathcal{E}}\circ\mathcal{U}_L).
\end{align}

\end{document}